\documentclass[12pt]{amsart}
\usepackage{geometry} 
\usepackage{graphicx}
\usepackage{amsfonts}
\usepackage{amsthm}
\usepackage{amsmath}
\usepackage{amssymb}
\usepackage[arrow,matrix,curve,color]{xy}
\usepackage{color}
\usepackage[normalem]{ulem}
\usepackage{tikz}
\usepackage{xspace}
\definecolor{light-blue}{rgb}{0.8,0.85,1}
\definecolor{light-red}{rgb}{1,.4,.4}
\definecolor{purp}{rgb}{.7,.3,1}
\definecolor{yel}{rgb}{1,1,.5}
\definecolor{cy}{rgb}{0,1,1}
\usepackage{colortbl}
\usepackage{multirow}
\usepackage{array}
\newtheorem{theorem}{Theorem}
\newtheorem{corollary}[theorem]{Corollary}
\newtheorem{lemma}[theorem]{Lemma}
\newtheorem{proposition}[theorem]{Proposition}
\theoremstyle{definition}
\newtheorem{example}[theorem]{Example}

\newtheorem{remark}[theorem]{Remark}

\newcommand{\co}{\colon\,}

\newcommand{\bT}{\mathbb T}
\newcommand{\bR}{\mathbb R}
\newcommand{\bC}{\mathbb C}

\newcommand{\bF}{\mathbb F}

\newcommand{\bZ}{\mathbb Z}
\newcommand{\bP}{\mathbb P}

\newcommand{\bQ}{\mathbb Q}

\newcommand{\cP}{\mathcal P}

\newcommand{\fp}{\mathfrak p}
\newcommand{\fg}{\mathfrak g}
\newcommand{\fk}{\mathfrak k}
\newcommand{\fu}{\mathfrak u}
\newcommand{\SO}{\mathop{\rm SO}}
\newcommand{\SU}{\mathop{\rm SU}}
\newcommand{\U}{\mathop{\rm U}}

\newcommand{\Sp}{\mathop{\rm Sp}}
\newcommand{\Spin}{\mathop{\rm Spin}}
\newcommand{\PSU}{\mathop{\rm PSU}}
\newcommand{\PSp}{\mathop{\rm PSp}}
\newcommand{\PSO}{\mathop{\rm PSO}}

\newcommand{\tG}{\widetilde G}
\newcommand{\tT}{\widetilde T}
\newcommand{\tR}{\widetilde R}

\newcommand{\pt}{\text{pt}}
\newcommand{\lp}{\textup{(}}
\newcommand{\rp}{\textup{)}}
\newcommand{\dual}{^{\vee}}
\newcommand{\Ext}{\operatorname{Ext}}
\newcommand{\Hom}{\operatorname{Hom}}

\newcommand{\Sq}{\operatorname{Sq}}

\newcommand{\diag}{\operatorname{diag}}
\newcommand{\image}{\operatorname{image}}
\newcommand{\lcm}{\operatorname{lcm}}

\title{Group dualities, T-dualities, and twisted $K$-theory}
\author{Varghese Mathai}
\address{Department of Pure Mathematics\\
School of Mathematical Sciences\\
University of Adelaide\\
Adelaide, SA 5005, Australia} 
\email[Varghese Mathai]{mathai.varghese@adelaide.edu.au}
\author{Jonathan Rosenberg}
\address{Department of Mathematics\\
University of Maryland\\
College Park, MD 20742-4015, USA} 
\email[Jonathan Rosenberg]{jmr@math.umd.edu}
\thanks{VM partially supported by
ARC grants DP130103924 and DP150100008.  
JR partially supported by NSF grants DMS-1206159 and DMS-1607162.  
JR also thanks VM and the University of Adelaide for their
hospitality during a research visit in January--March, 2016,
and thanks Jeff Adams for conversations about Section \ref{sec:orient} below.} 
  
\begin{document}
\begin{abstract}
This paper explores further the connection between Langlands
duality and T-duality for compact simple Lie groups, 
which appeared in work of Daenzer-Van Erp
and Bunke-Nikolaus.  We show that Langlands duality  gives
rise to isomorphisms of twisted $K$-groups, but that these
$K$-groups are trivial except in the simplest case of
$\SU(2)$ and $\SO(3)$.  Along the way we compute explicitly the
map on $H^3$ induced by a covering of compact simple Lie groups, 
which is either $1$ or $2$ depending in a complicated way
on the type of the groups involved.  We also give a new method
for computing twisted $K$-theory using the Segal spectral 
sequence, which is better adapted to many problems
involving compact groups. Finally we study a duality
for orientifolds based on complex Lie groups with an involution.
\end{abstract}

\maketitle

\section{Introduction}
\label{sec:intro}
This paper was motivated by a number of rather diverse
sources, especially the very intriguing
papers \cite{MR3285609,MR3361543}, which point out an interesting connection
between Langlands duality (which appears in representation theory
and number theory) and T-duality (a relationship between string theories
on different spacetime manifolds, especially when these are torus
bundles over a common base).  A second source comes from the extensive
literature on WZW (Wess-Zumino-Witten) models in physics, which
appear both as conformal field theories and as 
string theories (sigma models, to be precise) whose underlying spacetime manifold
is a Lie group, usually compact and simple.  In string theories
in general, D-brane charges are expected to take their values
in twisted $K$-theory of spacetime, so the study of WZW models
led to the study of twisted $K$-theory of compact Lie groups, first
by physicists (e.g., \cite{MR2079376,MR1877986,MR1834409,MR1960468,
MR2061550,MR2399311,MR2080884})
and later by mathematicians (Hopkins, unpublished, but quoted in
\cite{MR1877986}, and Douglas \cite{MR2263220}), although the computation 
of the twisted $K$-groups in the literature was mainly done for
compact {\em simple and simply connected} Lie groups.
The third motivation comes from the problem of trying to understand
exactly what one should expect from a ``duality theory'' for WZW models.
All three of these topics come together in the problems
discussed in this paper.

Section \ref{sec:Lang} deals with the connection between
{\em Langlands duality} and {\em T-duality}.  The papers
\cite{MR3285609,MR3361543}
show that on any compact simple Lie group $G$
not of type $B_n$ or $C_n$ ($n\ge 3$)\footnote{These
are the only cases where the flag manifolds for $G$ and
its Langlands dual are different.}, there is a class
$h\in H^3(G)$ for which the ``T-dual'' of the torus bundle
$G\to G/T$ is the Langlands dual $G\dual$ with a dual class
$h\dual\in H^3(G\dual)$. However, these papers give a description of
$h$ which is rather indirect, and 
and it is unclear how it relates to the generator of
$H^3(G)$.\footnote{The group $H^3(G)$ is infinite cyclic except for the
exceptional case of $\PSO(2n)$ for $n$ divisible by $2$; even in this case,
$H^3(G)$ has an infinite cyclic summand which is almost unique.}
An achievement of this paper is that in Sections
\ref{sec:Lang1} and \ref{sec:twistedKLang}, 
we are able to relate $h$ 
to an explicit multiple of the generator of $H^3(G)$, 
and also relate
$h\dual\in H^3(G\dual)$ to an explicit multiple of the generator
of $H^3(G\dual)$, thereby 
making an important clarification of the main result in the papers 
\cite{MR3285609,MR3361543}.
A key step towards this result is a novel explicit description 
of the pull-back map from $H^3(G)$ to $H^3(\tG)$, $\tG$ the
universal cover of $G$, in Section \ref{sec:fundclass},
Theorems \ref{thm:maponH3} and \ref{thm:PSO2n}.  The results
are surprisingly intricate, and must be done on a case-by-case basis.   

Having made the choice of $h$ and $h\dual$ more precise, we compute
the twisted $K$-theory $K^\bullet(G,h)$\,\footnote{This notation
  will always mean complex topological $K$-theory twisted by the
  class $h\in H^3$.} 
explicitly in our next central result, Theorem~\ref{thm:Tdualclasses},
whose proof is rather involved. The surprising upshot is
that the T-dualities coming from Langlands duality of simply connected
simple Lie groups \emph{never give any interesting isomorphisms of
  twisted $K$-groups}, except for the classic case of $\SU(2)$ and $\SO(3)$.
Section \ref{sec:twistedK}
then revisits the topic of computing twisted $K$-theory
$K^\bullet(G,h)$, for arbitrary choices of $h$.
This has been the subject of an extensive
literature, most notably \cite{MR1877986,MR2079376,
  MR2061550,MR2263220}, and the results are rather complicated
and hard to understand.  However, this is an important problem
because of the connection, discovered by physicists, between
these twisted $K$-groups and fusion rings and representations
of loop groups.  We therefore present in Section \ref{sec:twistedK}
an easier way of computing these twisted $K$-groups in some cases.
Theorems \ref{thm:twistedKofSU}, \ref{thm:twistedKofSp}, and
\ref{thm:twistedKofG2} recover some of the results of
\cite{MR2061550,MR2263220} via much more elementary methods, and
Theorems \ref{thm:twistedKPSU3} and \ref{thm:nonsc}
illustrate how the same methods
can be used in the case of non-simply connected Lie groups,
where the literature is still very incomplete.

Section \ref{sec:levelrank}
applies this to {\em level-rank duality}.
Finally, Section \ref{sec:orient} and
Theorem~\ref{thm:invduality} deal with a duality theory
for orientifolds based on complex Lie groups with a group
involution (either holomorphic or anti-holomorphic). It gives an isomorphism between the 
Real K-theories of a complex Lie group $G$ with respect to two different Lie-theoretic involutions.

\section{Langlands duality and the twisted $K$-theory of
  simple compact Lie groups}
\label{sec:Lang}

\subsection{The canonical class in $H^3$}
\label{sec:fundclass}

In this subsection, we consider what at first sight appears to
be a rather trivial and also somewhat esoteric question.
Suppose $G$ is a connected simple compact Lie group, 
not isomorphic to $\PSO(2n)$ with $n$ even, with
universal cover $\tG$. Since it is a classical fact that $\tG$
is $2$-connected, with $\pi_3(G)=\pi_3(\tG)\cong \bZ$, it follows
that $H^3(G)$ and $H^3(\tG)$ are both infinite cyclic, i.e., isomorphic
to $\bZ$.\footnote{Throughout the paper, homology and cohomology
  groups will be taken with integer coefficients unless
  specified otherwise.  The statement about $H^3(G)$ follows from
  considering the Serre spectral sequence for the fibration
  $\tG \to G\to K(\pi_1(G), 1)$ and using the fact that $\pi_1(G)$ is cyclic,
  so that we know its integral cohomology.  There is one exceptional case:
  if $G=\PSO(2n)$ with $n$ even, then $\pi_1(G)\cong (\bZ/2)\times(\bZ/2)$
  and the same spectral sequence shows that $H^3(G)\cong \bZ\oplus (\bZ/2)$.
  The torsion generator of $H^3(G)$ is killed on passage to the double cover
  $\SO(2n)$.}
In fact, $H^3(G)$ and $H^3(\tG)$ have canonical orientations (choices
of generator), since for a connected simple compact Lie group,
$H^3_{\text{deR}}(G)\cong H^3(G)\otimes \bR$ can be identified
with $H^3(\fg)$, the Lie algebra cohomology in degree $3$, which is
generated by the $3$-form $(X,Y,Z)\mapsto \langle X, [Y,Z]\rangle$, where
$\langle \underline{\phantom{X}}, \underline{\phantom{X}}\rangle$ is the
negative of the Killing form.
Let $\pi\co\tG\to G$ be the covering map; its degree $d$
divides the determinant of the Cartan matrix of $G$, and is equal to
it in case $G=\tG/Z(\tG)$ is the adjoint group for its Lie algebra.
Thus if $G$ is of adjoint type, $d$ is $n+1$ for type $A_n$ ($\tG=\SU(n+1)$),
$2$ for type $C_n$ ($\tG=\Sp(n)$), $2$ for type $B_n$ ($G=\SO(2n+1)$),
$4$ for type $D_n$ ($G=\PSO(2n)$, $n\ge 3$), $1$ for types $G_2$, $F_4$,
and $E_8$, $3$ for type $E_6$ and $2$ for type $E_7$.

The fact that $H^3(G)$ and $H^3(\tG)$ are both infinite cyclic
means that $\pi^*\co\! H^3(G)\to H^3(\tG)$ sends a generator of
$H^3(G)$ to some positive multiple of a generator of $H^3(\tG)$.
What multiple is this?  The following theorem largely answers this question.
Note incidentally that this question was taken up in
\cite[Appendix 1]{MR946997}, where it was pointed out that
there is a natural exact sequence
\[
0\to H_3(\tG) \xrightarrow{\pi_*} H_3(G) \to \pi_1(G) \to 0.
\]
However, this doesn't completely settle things since there is
an extension problem --- it is not clear when this exact sequence
splits, or when it ``partially splits.''

\begin{theorem}
\label{thm:maponH3}
Let $G$ be a connected simple compact Lie group of rank $n$, with
universal cover $\tG$, with $G$ not isomorphic to $\PSO(2n)$ 
with $n$ even, and let $d$ be the degree of the covering
$\pi\co\tG\to G$.  Think of $\pi^*\co H^3(G)\to H^3(\tG)$ as a map
$\bZ\to \bZ$.  
\begin{enumerate}
\item If $G$ is of type $A_n$ and if $n+1=2q$ with $q$ odd, then $\pi^*$ is
multiplication by $2$ if $d$ is even and the identity if $d$ is odd.
If $n+1$ is either odd or divisible by $4$,
then $\pi^*$ is the identity.
\item If $G$ is of type $C_n$ and $\pi$ is not
an isomorphism, i.e., $\tG=\Sp(n)$ and $G=\PSp(n)$, then $\pi^*$
is the identity if $n$ is even and is multiplication by $2$ if $n$ is odd.
\item If $G=\SO(2n)$ and $\tG=\Spin(2n)$ with $n\ge 3$, 
or if $G=\SO(2n+1)$ and
$\tG=\Spin(2n+1)$ with $n\ge 2$, or if $G=\PSO(2n)$ and
$\tG=\Spin(2n)$ with $n\ge 3$ odd
{\lp}so $G$ is of type $D_n$ or $B_n$,
though we are excluding one other possibility for $G$ in type
$D_n$ when $n$ is even{\rp}, then $\pi^*$ is the identity.
\item If $G$ is the adjoint group of $E_6$, then $\pi^*$ is
the identity, but if $G$ is the adjoint group of $E_7$, then $\pi^*$ is
multiplication by $d=2$.
\end{enumerate}
\end{theorem}
\begin{proof}
First recall that (by \cite[Ch. 4]{MR505692}, for example) there is a
transfer map $\pi_*\co H^3(\tG)\to H^3(G)$ such that $\pi_*\circ
\pi^*$ is multiplication by $d$.  So if $d$ is prime, that means there
are only two possibilities: either $\pi^*$ is multiplication by $d$
or else it is the identity.

(1) One case is obvious --- if $G$ has
rank $1$, so that $\tG\cong \SU(2)\cong \Sp(1)\cong \Spin(3)$,
then the only nontrivial covering map is $\pi\co \SU(2)\to \SO(3)$,
and this can be identified with the $2$-to-$1$ covering map $S^3\to
\bR\bP^3$. But a covering map of connected compact oriented $k$-manifolds
induces multiplication by the degree of the covering
on $H^k$, so in this case $\pi^*$ sends the generator of $H^3(\bR\bP^3)$ to
\emph{twice} a generator of $H^3(S^3)$.

For the remaining cases, observe that the $(n+1)$-fold covering map $\SU(n+1)
\to \PSU(n+1)$ factors through $G$.  So when $\pi^*=1$ for the adjoint
group $\PSU(n+1)$, then the same will be true for $G$ as well.
Next suppose that $n+1=2q$ with $q$ odd and we know that $\pi^*=2$
when $G=\PSU(n+1)$.  Of the two coverings $\SU(n+1)\to G$
and $G\to \PSU(n+1)$, exactly one is of even degree, and the $\pi^*$ maps for
these two coverings must multiply together to give $2$.  Since $\pi^*$
for any covering
divides the order of the covering, the only possibility is that $\pi^*$
is $2$ for the even covering and $1$ for the odd covering.  So the result
is reduced to the case $G=\PSU(n+1)$.

Now we apply a result of Baum and Browder \cite[Corollary 4.4]{MR0189063}:
if $f\co \U(n+1)\to \PSU(n+1)$ is the quotient map (recall that
$\PSU(n+1) = \U(n+1)/Z(\U(n+1))$), then $f^*$, the induced map
on cohomology with $\bF_p$ coefficients, for any prime $p$ dividing $n+1$,
has kernel equal to the ideal generated by the generator $y$ of
$H^2(\PSU(n+1),\bF_p)$.  Also note that the diagram
\[
\xymatrix{
\SU(n+1) \ar[r]^\iota \ar[rd]^\pi & \U(n+1)\ar[d]^f\\ & \PSU(n+1)}
\]
commutes, and since $\iota^*\co H^3(\U(n+1))\to H^3(\SU(n+1))$
is an isomorphism, this identifies the kernel of
$\pi^*\co H^3(\PSU(n+1),\bF_p)\to H^3(\SU(n+1),\bF_p)$
as well.

We recall from \cite[Th\'eor\`eme 11.4]{MR0064056} and \cite[Corollary 4.4]{MR0189063}
that if $n+1=p^r m$, with $p$ prime and $\gcd(m, p)=1$, then
\begin{equation}
  H^\bullet(\PSU(n+1),\bF_p)\cong \bF_p[y]/(y^{p^r})\otimes \textstyle{\bigwedge}
  (x_1, x_3, \cdots, \widehat{x_{2p^r-1}},\cdots, x_{2n+1}),
\label{eq:cohomPSU}
\end{equation}
with $\beta x_1 = y$ ($\beta$ the Bockstein) and with the $x_j$'s 
except for $x_1$ all reductions of integral classes.  If $p=2$ and $r=1$, one has
the additional relation $x_1^2=y$.  So, in all cases,
\[
H^3(\PSU(n+1),\bF_p) = \bF_p\, x_1y + \bF_p\, x_3\text{ (if present)}.
\]
We have $\beta(x_3)=0$ and $\beta(x_1y)=y^2$.  

There are a few cases to consider. 
If $p$ is odd, then $p^r \ge 3$ and $2p^r - 1 > 3$.
So $x_3$ is not omitted,
$y^2 \ne 0$, and $x_1y$ is not the reduction of an integral class.
Hence no integral class in $H^3(\PSU(n+1))$ reduces to something
nonzero
in the kernel of $f^*$ on mod $p$ cohomology.  Hence $\pi^*\co H^3(\PSU(n+1))
\to H^3(\SU(n+1))$ is never divisible by an odd prime.

Next, suppose $n+1=2q$ with $q$ odd.  If we choose $p=2$, then $r=1$,
so $x_3$ is missing in \eqref{eq:cohomPSU}, and $y^2=0$.  So
$x_1y$ is the reduction of the generator of $H^3(\PSU(n+1))$ in this
case, and it lies in the kernel of $f^*$.  Hence $\pi^*$ is divisible by $2$.
Since we have already shown it is not divisible by any odd prime, it is
exactly $2$.

Finally, consider the case $p=2$ when $n+1$ is divisible by $4$. In
this case, $p^r\ge 4$ so, again, $2p^r - 1 > 3$.  So $x_3$ is present
in \eqref{eq:cohomPSU}, and $y^2\ne 0$, so that $x_1y$ is 
not the reduction of an integral class.  Hence the kernel of $f^*$
does not contain the nonzero reduction of any integral cohomology
class, and $\pi^*$ cannot be divisible by $2$.  Since it is also not
divisible by any odd prime, $p^*=1$ in this case.

Note that since $A_3=D_3$, $\PSU(4) = \PSO(6)$. So
$\pi^*=1$ in the case $\tG=\Spin(6)$, $G=\PSO(6)$. 

(2) The case $G=\PSp(n)$, $\tG=\Sp(n)$, $n\ge2$, is
handled using \cite[Th\'eor\`eme 11.3]{MR0064056},
which gives
\begin{equation}
  H^\bullet(\PSp(n),\bF_2)\cong \bF_2[a]/(a^{4s})\otimes \textstyle{\bigwedge}
  (x_3, \cdots, \widehat{x_{4s-1}},\cdots, x_{4n-1}),
\label{eq:cohomPSp}
\end{equation}
with $a$ of degree $1$ and $s$ the largest power of $2$ dividing $n$.
Note that $4s=4$ if $n$ is odd and $4s \ge 8$ if $n$ is even.
The Bockstein $\beta=\Sq^1$ sends $a$ to $a^2$, and
all torsion is of order $2$, so the integral
cohomology is
\begin{equation}
  H^\bullet(\PSp(n),\bZ)\cong\bZ[a^2]/(2a^2,a^{4s})\otimes 
  \textstyle{\bigwedge}
  (x_3, \cdots, x_{4n-1}),
\label{eq:cohomPSpZ}
\end{equation}
with $x_{4s-1}$ reducing mod $2$ to $a^{4s-1}$.
Consider the commuting diagram of fibrations
\begin{equation}
\xymatrix{\Sp(n-1)\ar[r] \ar@{=}@<1ex>[d]& \Sp(n) \ar[r] \ar@{->>}[d]^\pi
  & S^{4n-1}\ar@{->>}[d]\\
  \Sp(n-1)\ar[r] & \PSp(n)  \ar[r]&  \,\bR\bP^{4n-1}.}
\label{eq:fibs1}
\end{equation}
From \eqref{eq:cohomPSpZ}, $H^4(\PSp(n))$ vanishes if and only if $s=1$,
i.e., if and only if $n$ is odd.  However, in the spectral sequence
for the bottom row of \eqref{eq:fibs1}, $E_2^{4,0}=\bF_2\,a^4$.
So $d_4(x_3)$ will be $0$ if $n$ is even and must equal $a^4$ if
$n$ is odd.  Comparing with the spectral sequence for the top row of
\eqref{eq:fibs1}, we see that we get a diagram
\[
\xymatrix{
  H^3(\Sp(n-1)) \ar@{=}[r] \ar@{=}[d]
  & E_\infty^{0,3} & \ar[l]^(.6)\cong H^3(\Sp(n))\\
  H^3(\Sp(n-1)) \ar@{=}[r] & E_\infty^{0,3} & \ar[l]^(.6)\cong H^3(\PSp(n))
\ar[u]^{\pi^*}
}
\]
for $n$ even, whereas for $n$ odd, the image under $\pi^*$ of $H^3(\PSp(n))$
in $H^3(\Sp(n-1))\cong H^3(\Sp(n))$ is of index two.  This is consistent,
incidentally, with what happened for $\Sp(1)=\SU(2)$ and for
$\PSp(1)=\SO(3)$.

(3) Now consider the orthogonal and spinor groups.
Here we proceed by induction on the dimension.
Recall that $\Spin(4)$ is not simple and
splits as a product of two copies of $\SU(2)$, while $\Spin(3)=\SU(2)$
was already dealt with above, so we start the induction with $\Spin(5)$.
This is the same as $\Sp(2)$ and $\SO(5)$ is the same as $\PSp(2)$,
so we've already handled this case.  We could get the case of $\Spin(6)$
by the induction below, but instructive to give a separate argument
just as a check.  Indeed, $\Spin(6)=\SU(4)$,
and its cohomology is torsion-free, though $\SO(6)$ is a double cover of
$\PSU(4)$. But as in \eqref{eq:fibs1}, we have a commutative diagram of
fibrations
\[
\xymatrix{\SU(3)\ar[r] \ar@{=}@<1ex>[d]& \SU(4) \ar[r] \ar@{->>}[d]^\pi
  & S^{7}\ar@{->>}[d]\\
  \SU(3)\ar[r] & \SO(6) \ar[r]&  \,\bR\bP^7.}
\]
(The action of $SO(6)$ on $\bR\bP^7$ is by dividing out the
action of $\SU(4)$ on $S^7$ by $\pm 1$.)  As before, we get a commuting
diagram of Serre spectral sequences.  Along the vertical axis we have
the cohomology ring of $\SU(3)$, which is $\bigwedge(x_3,x_5)$.
Along the horizontal axis we have $\bigwedge(x_7)$ in the case
of $\SU(4)$ and $H^\bullet(\bR\bP^7)$, generated by a torsion-free
class $x_7$ with $x_7^2=0$ and a $2$-torsion class $y_2$
with $y_2^4=0$, in the case of $\SO(6)$.
The spectral sequence for $\SU(4)$ collapses,
while the one for $\SO(6)$ has a potential differential $d_4$ sending
the generator $x_3$ in $E_2^{0,3}$ to $y_2^2$.  However, this
would kill off the torsion in $H^4(\SO(6))$, and Borel proves
in \cite[Th\'eor\`eme 8.6]{MR0064056} that $H^\bullet(\SO(n))$
has $2$-torsion in degree $4$ as long as $n\ge 5$.  So the differential
$d_4$ must vanish on $x_3$ and we conclude that $\pi^*$ is an isomorphism
via a diagram chase as before.

Now we're ready for the inductive step.  Assume $n\ge 6$ and
we already know
\[\pi^*\co H^3(\SO(n-1))\to H^3(\Spin(n-1))
\]
is an isomorphism, and consider
the commutative diagram of fibrations
\[
\xymatrix{\Spin(n-1)\ar[r] \ar@{->>}@<1ex>[d]^\pi& \Spin(n) \ar[r] \ar@{->>}[d]^\pi
  & S^{n-1}\ar@{=}[d]\\
  \SO(n-1)\ar[r] & \SO(n) \ar[r]&  \,S^{n-1}.}
\]
Examination of the Serre spectral
sequences shows that $H^3=E_\infty^{0,3}$, for both $\SO(n)$ and
$\Spin(n)$.  By inductive hypothesis, $\pi^*$ is an isomorphism on
$E_2^{0,3}$, and there is no room for differentials affecting the
vertical axis until we get to $H^4$ (since the non-zero columns
are distance $n-1\ge 5$ apart.)  So $E_2^{0,3}=E_\infty^{0,3}$
in both spectral sequences and we're done.

To get the case of $G=\PSO(2n)$ with $n$ odd, recall that we already
did this case when $n=3$, since $\PSO(6)=\PSU(4)$. However,
by \cite[Corollary 1.8]{MR0372899}, the inclusion $\PSO(6)\to
\PSO(2n)$ induces an isomorphism on cohomology in degree $\le 3$
for any $n$ odd.  Similarly for the inclusion  $\Spin(6)\to \Spin(2n)$,
by what we've already done.  So the case of $\PSO(2n)$ follows
from the known case of $\PSO(6)$.

(4) To handle the case where $G$ is the adjoint group of $E_6$, since
$d=3$ in this case, we need to study the cohomology of $G$ and $\tG$ with
$\bF_3$ coefficients.  Fortunately this is known. By
\cite{MR591603} or by \cite{Fung},
\begin{equation}
H^\bullet(\tG,\bF_3) = \bF_3[e_8]/(e_8^3)\otimes \bigwedge
(e_3,e_7,e_9,e_{11},e_{15},e_{17}),
\label{eq:E6}
\end{equation}
where $e_7=\cP^1e_3$, $e_8=\beta e_7$, and $e_{15}=\cP^1 e_{11}$,
whereas by \cite{MR0474355},
\begin{equation}
H^\bullet(G,\bF_3) = \bF_3[e_2,e_8]/(e_2^9,e_8^3)\otimes \bigwedge
(e_1,e_3,e_7,e_9,e_{11},e_{15}).
\label{eq:AdE6}
\end{equation}
(One can also read off \eqref{eq:E6} and 
\eqref{eq:AdE6} from \cite{MR1989490}.)
To understand what's going on here, the comments in \cite{Fung} are
useful.  To a good degree of approximation, $\tG$ looks like $K(\bZ,3)$,
whose cohomology is known and agrees with \eqref{eq:E6} in low degree.
Now $G$, even though it is not simply connected, is a simple space
and so has a Postnikov tower, the bottom of which looks like a fibration
\begin{equation}
\xymatrix{K(\bZ, 3) \ar[r] & X\ar[d]\\ & \,K(\bZ/3, 1).}
\label{eq:PostAdE6}
\end{equation}
The $k$-invariant of \eqref{eq:PostAdE6}  is identified with
the differential $d_4$ of the associated Serre spectral sequence,
taking the canonical class in $H^3(K(\bZ, 3))$ to
the $k$-invariant living in $H^4(K(\bZ/3, 1)) = \bF_3 e_2^2$,
with notation consistent with \eqref{eq:AdE6}.  Since $e_2^2\ne 0$ in
$H^4(G,\bF_3)$, the $k$-invariant must vanish, meaning that
$X\simeq K(\bZ, 3) \times K(\bZ/3, 1)$.  Indeed, the low-dimensional
cohomology of $G$ indeed looks like that of $K(\bZ, 3) \times K(\bZ/3, 1)$.
So to compute $\pi^*$, we can replace  $\tG$ by $K(\bZ, 3)\times E(\bZ/3)$
and $G$ by $K(\bZ, 3)\times B(\bZ/3)$ without changing the low-degree
cohomology.  Since $H^3$ lives on the $K(\bZ, 3)$ factor, we see
that $\pi^*$ is the identity on cohomology with coefficients in $\bF_3$
and thus on integral cohomology as well.

The case of the adjoint group $G$ of $E_7$ is handled similarly,
even though the eventual result is different.  This
time $\pi$ is a covering of degree $2$ and so we have to look at
cohomology with coefficients in $\bF_2$.  For $\tG$ this turns
out to be (\cite[Proposition 5.1]{MR0415651} or
\cite[Proposition 2.30]{Fung}):
\begin{equation}
H^\bullet(\tG,\bF_2) = \bF_2[x_3,x_5,x_9]/(x_3^4,x_5^4,x_9^4)\otimes \bigwedge
(x_{15},x_{17},x_{23},x_{27}).
\label{eq:E7}
\end{equation}
And for the adjoint group $G$ we get (\cite[Theorem 5.3]{MR0415651}
or \cite{MR509497}):
\begin{equation}
  H^\bullet(G,\bF_2) = \bF_2[x_1, x_5, x_9] / (x_1^4 ,x_5^4,x_9^4)
  \otimes \bigwedge (x_6,x_{15},x_{17},x_{23},x_{27}).
\label{eq:AdE7}
\end{equation}
We again have an approximation of $G$ by a two-stage Postnikov system
\begin{equation}
\xymatrix{K(\bZ, 3) \ar[r] & X\ar[d]\\ & \,K(\bZ/2, 1).}
\label{eq:PostAdE7}
\end{equation}
The $k$-invariant of \eqref{eq:PostAdE7}  is identified with
the differential $d_4$ of the associated Serre spectral sequence,
taking the canonical class $x_3$ in $H^3(K(\bZ, 3))$ to
the $k$-invariant living in $H^4(K(\bZ/2, 1)) = \bF_2 x_1^4$,
with notation consistent with \eqref{eq:AdE7}. But since
$x_1^4 =0$ in $H^\bullet(G,\bF_2)$, this time the
$k$-invariant is non-zero, meaning that $d_4(x_3) = x_1^4$.
The edge homomorphism $H^3(X,\bF_2)\to H^3(K(\bZ,3),\bF_2)$
is therefore $0$. But this map must also coincide with the map
$\pi^*\co H^3(G,\bF_2)\to H^3(\tG,\bF_2)$
(since $K(\bZ,3)$ is a good approximation to $\tG$
and $X$ is a good approximation to $G$), so we conclude that
$\pi^*$ is \emph{not} the identity on $H^3$ with coefficients in $\bF_2$,
hence is not an isomorphism on integral cohomology.
\end{proof}
\begin{remark}
\label{rem:SUn}
For $G$ of types $G_2$, $F_4$, or $E_8$,  we automatically have
$d=1$ and there is nothing to prove. So more
work remains only in certain cases in type $D_n$ ($n$ even).
Just to fill in a little more about the case of $D_n$ ($n$ even),
we have the following. 
\end{remark}  
\begin{theorem}
\label{thm:PSO2n}
Let $G=\PSO(2n)$ with $n$ even, $n\ge 4$,
and let $\pi\co \tG=\Spin(2n)\to G$ be its
universal covering map.  Recall that in this case that $H^3(G)\cong \bZ \oplus
(\bZ/2)$.  Then $\pi^*\co H^3(G)\to H^3(\tG)$ kills the torsion and on
the torsion-free part is an isomorphism for $n$ divisible by $4$,
multiplication by $2$ for $n\equiv 2 \pmod 4$.
\end{theorem}
\begin{proof}
Since we have already shown in Theorem \ref{thm:maponH3}(3) that 
the covering $\Spin(2n)\to \SO(2n)$ induces an isomorphism on $H^3$,
it suffices to study the behavior of the covering map $\SO(2n)\to \PSO(2n)$
on the torsion-free part of $H^3(\PSO(2n))$.  For this we use the
fibration
\begin{equation}
\SO(2n) \to \PSO(2n) \to \bR\bP^{\infty}
\label{eq:fibPSO}
\end{equation}
together with Borel's calculation of $H^\bullet(\SO(2n),\bF_2)$ in
\cite[Th\'eor\`eme 8.7]{MR0064056} and of $H^\bullet(\PSO(2n),\bF_2)$
in \cite[Th\'eor\`eme 11.5]{MR0064056}.  In low degrees, the $E_2$
term for the Serre spectral sequence of \eqref{eq:fibPSO} 
looks like Figure \ref{fig:SSSPSO}.  (Recall that we are assuming $n\ge 4$.)
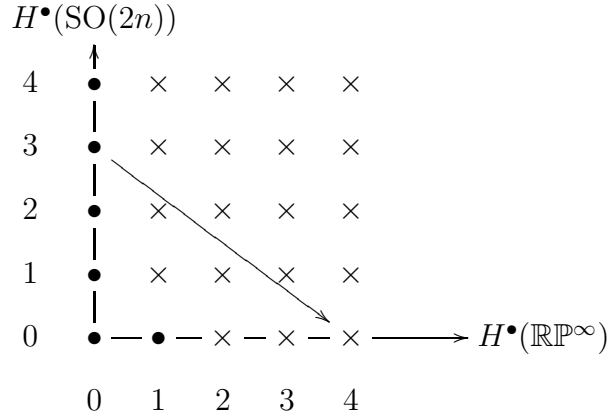
\begin{figure}[hbt]
\[
\xymatrix@!0{
&H^\bullet(\SO(2n))&&&&&&&\\
4&\bullet\ar[u]&\times&\times&\times&\times&&&\\
3&\bullet\ar@{-}[u]\ar@{-}[u]\ar@[red][rrrrddd]
&\times&\times&\times&\times&&&\\
2&\bullet\ar@{-}[u]&\times&\times&\times&\times&&&\\
1&\bullet\ar@{-}[u]&\times&\times&\times&\times&&&\\
0&\bullet\ar@{-}[r] \ar@{-}[u] &\bullet\ar@{-}[r] &\times\ar@{-}[r]&\times\ar@{-}[r]&\times\ar[rrr]&&& H^\bullet(\bR\bP^{\infty})\\
&0&1&2&3&4&&&
}\]
\caption{The Serre SS for $\bF_2$-cohomology of $\PSO(2n)$.
  The symbol $\bullet$ denotes a simple generator ($h_j$ on
  the vertical axis, $a$ on the horizontal axis).  The diagonal
  arrow shows the transgression $d_4$ for the case $n\equiv 2\pmod 4$.}
\label{fig:SSSPSO}
\end{figure}
Borel proves in \cite[Th\'eor\`eme 8.7]{MR0064056} that
\[
H^\bullet(\SO(2n),\bF_2)\cong
\left\langle h_1, h_2, \cdots, h_{2n-1}\right\rangle,
\]
with $h_j^2 = 0$ if $2j > 2n-1$, $h_j^2 = h_{2j}$ if $2j < 2n-1$,
and with $\Sq^1h_j=0$ exactly when $j$ is odd and $\ge 3$.
Thus $H^3(\SO(2n),\bF_2)$ is spanned by $h_1h_2=h_1^3$ and $h_3$,
but the former is not the reduction of an integral class,
whereas $h_3$ is the mod $2$ reduction of the generator of
$H^3(\SO(2n),\bZ)$.  Since the map $H^3(\PSO(2n)) \to H^3(\SO(2n))$
can only be either an isomorphism or multiplication by $2$ on
the torsion-free part of $H^3(\PSO(2n))$, it suffices to determine
whether $h_3$ on the vertical axis survives to $E_\infty$.  If it
does, that means the map is an isomorphism on the torsion-free
summand, and if it doesn't, that means the map is multiplication
by $2$ on the torsion-free summand.

Now by \cite[Th\'eor\`eme 11.5]{MR0064056}, 
\[
H^\bullet(\PSO(2n),\bF_2) \cong
\left(\bF_2(a)/(a^s)\right)\otimes A,\quad A=\left\langle x_1, x_2, \cdots, \widehat{x_{s-1}}, \cdots, x_{2n-1}\right\rangle,
\]
where $a$ is in degree $1$,
$s\ge 4$ is the largest power of $2$ dividing $2n$, and
in the algebra $A$, $x_1^2=x_2$, $x_2^2=x_4$, etc.,
and the monomials in the $x$'s with no $x_j$ repeated are a basis for
$A$ over $\bF_2$. It is easy to see that in terms of the picture in
Figure \ref{fig:SSSPSO}, $a$ has to correspond to the generator
of $H^1(\bR\bP^\infty,\bF_2)$, and $x_j$ should map to $h_j$ in
most cases.

If $n\equiv 0\pmod 4$, then $s\ge 8$
and $a^4\ne 0$, so $h_3$ on the vertical axis has to survive to $E_\infty$.
However, when $n\equiv 2\pmod 4$, then $s = 4$ and there is no $x_3$
class in $H^3(\PSO(2n),\bF_2)$, while $a^4=0$.  This means there has
to be a transgression arrow as in Figure \ref{fig:SSSPSO}
sending $h_3$ on the vertical axis to $a^4$ on the horizontal axis.
In this case, $H^3(\PSO(2n),\bF_2)$ is spanned by $a^3$ and
$ax_2 + a^2x_1$, which are reductions of integral classes, as well as
$ax_2$ and $x_1x_2=x_1^3$, which have nontrivial Bocksteins and thus
cannot come from integral classes.  And none of these map to $h_3$.
That concludes the proof.
\end{proof}

\subsection{T-dualities from Langlands duality}
\label{sec:Lang1}

Now we return to the main subject of Section \ref{sec:Lang}, the connection
between Langlands duality and T-duality.  We make no claims for originality
here --- we are just restating the results of \cite{MR3285609} and \cite{MR3361543}
in a form that will be convenient for later calculations.  Let $G$ be a connected compact
simple Lie group, and let $T$ be a maximal torus in $G$.  Associated to $T$
we have the \emph{weight lattice}, the character group $\Hom(T,\bT)$ of $T$,
which can also be identified with $H^1(T)$, and the \emph{coweight lattice}
$H_1(T)$, which can be identified with the set of
cocharacters $\Hom(\bT,T)$.  The weight lattice is also the set of allowable weights 
for irreducible representations of $G$.  Inside the weight lattice is the set 
of \emph{roots} $\Phi(G,T)$ (these are exactly the non-zero
weights of the adjoint representation), which span a sublattice, called the 
\emph{root lattice}.  The weight lattice and the root lattice coincide 
exactly when $G$ is an adjoint group.  The
triple $(H^1(T), H_1(T), \Phi(G,T))$ determines $G$ up to isomorphism.
The \emph{Langlands dual} $G\dual$ of $G$ is the connected compact
simple Lie group obtained by interchanging $H^1(T)$ and $H_1(T)$ and
replacing the roots $\Phi(G,T)$ by the \emph{coroots} $\Phi(G,T)\dual$.
Langlands duality performed twice brings one back to the starting point,
and the duality interchanges simply connected groups and adjoint groups.
The Langlands dual $G\dual$ is locally isomorphic to $G$ if $G$ is of
type A, D, E, F, or G, while Langlands duality interchanges types B and C.
The group $\SO(2n)$ of type $D_n$ is neither simply connected nor
an adjoint group, and is in fact self-dual.
See \cite[\S2]{MR3285609} for a very clear explanation of all of this, as
well as for additional references.  The following is the main theorem
of \cite{MR3285609} and \cite{MR3361543} (though the case of
type $B_2=C_2$ is never mentioned there explicitly).

\begin{theorem}[Daenzer-van Erp, Bunke-Nikolaus]
Let $G$ be a connected compact simple Lie group of rank $n$
with maximal torus $T$, and let $G\dual$
be its Langlands dual.  Assume $G$ is not of type $B_n$ or $C_n$
with $n\ge 3$; this guarantees that $G$ and $G\dual$ are locally
isomorphic {\lp}i.e., have isomorphic Lie algebras{\rp}.
Then there are canonical choices $h\in H^3(G)$
and $h\dual\in H^3(G\dual)$ for which $(G,h)$ and $(G\dual,h\dual)$
are T-dual as principal $\bT^n$-bundles over the flag manifold
$G/T$, in a sense which we will make precise below.
\label{thm:BN}
\end{theorem}
\begin{remark}
\label{rem:BN}
Theorem \ref{thm:BN} actually requires some further interpretation, because
the term ``T-duality'' has various meanings in the literature.  For example the
definition used in \cite{MR3285609} is much weaker than the one used in \cite{MR3361543},
which comes from \cite{MR2287642}.  One could also work just as well with
the definition in \cite{MR2222224}.  However, the definition used in
\cite{MR3285609} is not strong enough for getting the results on twisted $K$-theory
which are our main interest here.  Just to fix terminology for what follows,
we can phrase the definition of T-duality as follows.  A set of
\emph{T-duality data} over a
a space $X$ begins with a pair $(E,h)$ consisting of a principal $\bT^n$-bundle
$p\co E\to X$ and a class $h\in H^3(E)$.  However, this by itself is not
enough for a (classical) T-dual to be defined; one needs two additional things.
First, the class $h$ must pull back to $0$ on each torus fiber.  In the
case where $E=G$ is a connected compact simple Lie group and $X=G/T$,
this condition is automatic \cite[Corollary 8.3]{MR3361543}.\footnote{For another proof,
note that it suffices to check this for cohomology with coefficients in
$\bR$.  But $H^3(G,\bR)\cong H^3_{\text{deR}}(G,\bR)$ is generated
by the $3$-form $(X,Y,Z)\mapsto \langle X, [Y, Z]\rangle$,
$X,Y,Z\in \fg$ (the Lie algebra) and $\langle\ ,\ \rangle$ the Killing form.
Since the bracket vanishes on the Lie algebra of $T$, this form
restricts to $0$ in $H^3_{\text{deR}}(T,\bR)$.} Then to get
a well-defined T-dual, one needs a choice of a lifting $X\to \widetilde R$ of the
classifying map $X\to R$ of the pair $(E,h)$ over $X$.  The classifying spaces
$R$ and $\widetilde R$ are discussed in \cite[\S5]{MR2222224}.  The space
$\widetilde R$ is a two-stage Postnikov system
\begin{equation}
\xymatrix{K(\bZ,3)\ar[r] & \widetilde R\ar[d] \\ & K(\bZ^n,2) \times K(\bZ^n,2) ,}
\label{eq:classsp}
\end{equation}
with $k$-invariant $x_1\,y_1+\cdots+x_n\,y_n$ (where $x_j$ and $y_j$
are the canonical classes in $H^2$ of the first and second $K(\bZ^n,2)$ 
factors, respectively),
and $R$ is a more complicated space that has $\widetilde R$ as its
universal covering. In \eqref{eq:classsp}, the bundle projection projected
into the first factor corresponds to the characteristic class $c(p)$ of the torus
bundle $E\xrightarrow{p} X$, while the bundle projection projected
into the second factor corresponds (when $\dim T=1$ --- this case is
easier to explain) to the push-forward of the class $h$
under the Gysin map $p_!\co H^3(E) \to H^2(X)$.  
T-duality interchanges these.  In our situation,
$X=G/T$ will be simply connected, so existence of a lift from $R$ to
$\widetilde R$ is automatic, but we will see that there is a canonical lift.

Finally, there is one additional complication.  In the case of Langlands duality,
$G$ and $G\dual$ are really torus bundles over two different spaces,
$G/T$ and $G\dual/T\dual$.  If $G$ is of type $B$ or $C$ with rank $n\ge 3$,
then these spaces are not even homotopy equivalent, although they are
homotopy equivalent after inverting $2$ \cite{MR0131278}.  However
in all cases except these, there is a diffeomorphism $G/T\to G\dual/T\dual$
coming from an isomorphism of root systems $\Phi\to \Phi\dual$.
(This diffeomorphism can be taken to be the identity in types A, D, and E,
but it can't be the identity in types $B_2=C_2$, $G_2$, or $F_4$,
since long roots have short coroots and \emph{vice versa}.)
\end{remark}

To finish making Theorem \ref{thm:BN} precise, we need to make explicit the
choice of $h$ and the choice of the lift $X\to R$ of the classifying map
of $(G,h)$ over $X=G/T$.  Let $G$, $T$ be as in Theorem \textup{\ref{thm:BN}}, and let 
$\pi\co\tG\to G$ be the universal
cover of $G$, $\tT$ the inverse image of $T$ in $\tG$ {\lp}a maximal torus
in $\tG${\rp}.  Let $(G\dual, T\dual)$ be the Langlands dual of
$(G,T)$ and let $\pi\dual\co\widetilde{G\dual}\to G\dual$ be its universal cover,
$\widetilde{T\dual}$ the inverse image of $T\dual$ in $\widetilde{G\dual}$.
Then Theorem \textup{\ref{thm:BN}} holds for the following explicit
choice of $h$ and a lift.  Recall that $\tG$ is $2$-connected and that
$G/T\cong \tG/\tT$ is a complex projective algebraic variety with a cell decomposition
with only even-dimensional cells, indexed by elements of the Weyl group
$W=N_G(T)/T$.  Thus the cohomology of $G/T$ is torsion-free and concentrated
in even degrees, and the sum of the Betti numbers is $|W|$,
while $\tG$ has no cohomology in degrees $1$ and $2$ and
$H^3(\tG)\cong \bZ$.  The group $G$ in general has torsion in its cohomology
in degree $2$ but also has $H^3(G)\cong \bZ$. {\lp}The relationship between
$H^3(G)$ and $H^3(\tG)$ was discussed above in Section 
\textup{\ref{sec:fundclass}}.\footnote{As explained there, a slight
adjustment is needed in case $D_n$ for $n$ even when $G=\PSO(2n)$,
since in this case $H^3(G)\cong \bZ\oplus \bZ/2$.  The torsion
summand is harmless here.}{\rp}  In the Serre spectral sequence for the
fibration $\tT\to \tG\to \tG/\tT=G/T$, the differential $d_2$ gives a transgression 
isomorphism $\psi\co H^1(\tT) \xrightarrow{\cong} H^2(G/T)$.  {\lp}Restricted to
$H^1(T)$, its image has finite index, with the quotient identified to
$H^2(G)\cong \Ext(\pi_1(G),\bZ)$, 
via the universal coefficient theorem.{\rp}
Since $H^2(T)$ is naturally identified with $\bigwedge^2(H^1(T))$, and
similarly $H^2(\tT)$ with $\bigwedge^2(H^1(\tT))$, the differential $d_2$
sends $H^2(T)$, respectively $H^2(\tT)$, to $H^2(G/T)\otimes H^1(T)$
{\lp}resp., $H^2(G/T)\otimes H^1(\tT)${\rp} via $a\wedge b \mapsto
\psi(a)\otimes b - \psi(b)\otimes a$.

Now we give a slight reformulation of the recipe given 
in \cite[Theorem 8.5]{MR3361543}, in the notation above.
The principal $T$-bundle $p\co G\to G/T$ has a classifying map
$G/T\to BT$, determined up to homotopy by the induced map on
cohomology $H^2(BT)\to H^2(G/T)$\,\footnote{This is because
$BT$ is a $K(\bZ^n,2)$ space.}, and since $H^2(BT)$ is
canonically isomorphic to $H^1(T)$, we can think of this as
a characteristic class $c(p)\in \Hom(H^1(T), H^2(G/T))$.  
{\lp}We use the letter $c$ since this is precisely the 
Chern class when $n=1$.{\rp} Fix an isomorphism
$\widetilde{G\dual}\to \tG$ {\lp}possible since we are excluding
types B and C with $n\ge 3${\rp} and let $\phi\co H^2(G\dual/T\dual) \to
H^2(G/T)$ be the induced isomorphism.
\begin{theorem}[Bunke-Nikolaus]
\label{thm:BN1}
In this situation,
\[
\begin{aligned}
c(p)&=\psi\circ \pi^*\co H^1(T)\to H^2(G/T),\quad\text{\textup{and}} \\
\phi\circ c(p\dual)&=\phi\circ\psi\dual\circ (\pi\dual)^*\co
H^1(T\dual)\to H^2(G/T).
\end{aligned}
\]
Let $h\in H^3(G)$ be the image in $E_\infty^{2,1}\cong H^3(G)$
of the composite $\phi\circ c(p\dual)\co
H_1(T) = H^1(T\dual)\to H^2(G/T)$, viewed in
$\Hom(H_1(T), H^2(G/T))\cong H^2(G/T)\otimes H^1(T)$, 
which makes sense since this element is easily seen to be
killed by $d_2$.  Similarly for $h\dual \in H^3(G\dual)$.
Then $(G, h)$ together with the map $G/T\to \tR$
defined by the homotopy commutative diagram of fibrations
\[
\xymatrix{
T\ar@{^{(}->}[r] \ar@{=}[d] &G \ar[r]^p \ar[d]_{(h, \phi\circ c(p\dual))}& G/T\ar@{.>}[d] \\
T \ar[r] & K(\bZ, 3)\times K(\bZ^n,2)\ar[r] & \tR}
\]
{\lp}the dashed arrow is a homotopy push-out{\rp}
is a set of T-duality data, and the dual data corresponds to
$(G\dual, h\dual)$ viewed also as a bundle over $G/T$
via $\phi$.
\end{theorem}
\begin{proof}
The one thing to check is that $\phi\circ c(p\dual)$ is killed by
$d_2$ and thus gives a genuine class in $H^3(G)$. This involves
showing that the map is $W$-invariant.  But 
$S^2(H^1(T)\otimes_\bZ \bQ)$ contains a unique $W$-invariant
element up to scale, namely the Killing form, so we just need
to check that $\psi^{-1}\circ \phi\circ c(p\dual)\co H_1(T) \to H^1(\tT)$
is a multiple of this form, which is immediate
from the formula.

As explained in \cite{MR3361543}, it is then almost a tautology that
the T-dual of this data is the data for $(G\dual, h\dual)$, since
$G$ and $G\dual$ play completely symmetrical roles here.
\end{proof}

\begin{corollary}
\label{cor:twistedK}
In the setting of Theorems \textup{\ref{thm:BN}} and 
\textup{\ref{thm:BN1}},  one gets an isomorphism
of twisted $K$-groups 
\[
K^\bullet(G, h) \cong 
K^{\bullet+n}(G\dual, h\dual).
\]
\end{corollary}
\begin{proof}
This follows immediately from \cite{MR2287642}
or \cite{MR2222224}, applied to the T-duality.
\end{proof}

\subsection{Twisted $K$-theory for Langlands dual pairs}
\label{sec:twistedKLang}

The results of Section \ref{sec:Lang1}, especially Corollary
\ref{cor:twistedK}, provide new examples of isomorphisms
between twisted $K$-groups on different compact simple Lie groups.
This provides additional motivation for understanding the
computation of such twisted $K$-groups and computing exactly
what one gets in the case of the T-duality pairs coming from
Langlands duality.  We start with a few very explicit examples.

\begin{example}[{$\SU(2)$ and $\SO(3)$}]
    \label{ex:SU2}
We begin with the simplest case of $G=\SU(2)$, $G\dual = \SO(3)$,
the only dual pair in rank $1$ (though this pair has other aliases,
e.g., $G=\Sp(1)$ and $G\dual = \PSp(1)$, or $G=\Spin(3)$ and
$G\dual = \SO(3)$). In this case $G/T=S^2$,
$G$ is topologically $S^3$, $p$ is the Hopf bundle $S^3\to S^2$, and
$G\dual$ is topologically $\bR\bP^3$. It is easy to see that the Chern class
$c(p)=1$ (for the usual identification of $H^2(S^2)$ with $\bZ$)
and $c(p\dual)=2$. Since T-duality for circle bundles (see
for example \cite{MR2080959,MR2116734,MR2560910})
requires $p_!(h)=c(p\dual)$ and $(p\dual)_! =
c(p)$, we easily compute from the Gysin sequences that
$h=2$ and $h\dual=1$, again for the usual identifications of
the $H^3$ groups with $\bZ$.

Let us now check Corollary \ref{cor:twistedK} in this example.
We can compute the twisted $K$-theory via the Atiyah-Hirzebruch spectral
sequence (AHSS)
for twisted $K$-theory, first discussed in \cite{MR679694,MR1018964}
and further explained in \cite{MR2172633,MR2307274}.
Since $G$ and $G\dual$ are $3$-dimensional, the only differential
is $d_3$, and $\Sq^3=0$ for both $G$ and $G\dual$.  So (for any
choice of $h$ and $h\dual$), the twisted $K$-groups are computed
from the spectral sequence shown in Figure \ref{fig:AHSSSU2}.
(By Bott periodicity, the differentials repeat vertically with period $2$.)

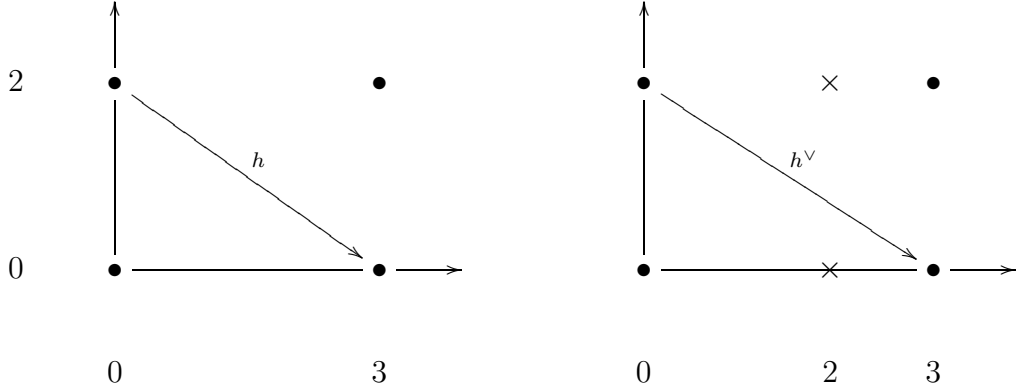
\begin{figure}[hbt]
\[
\xymatrix{
&&&&&&&&&&&\\
2&\bullet \ar[ddrrr]^h \ar[u]&&&\bullet& &&\bullet \ar[ddrrr]^{h\dual} \ar[u]&&\times&\bullet&\\
&&&&&&&&&&& \\
0&\bullet\ar@{-}[rrr] \ar@{-}[uu] &&&\bullet\ar[r] &&
&\bullet\ar@{-}[rrr] \ar@{-}[uu] &&\times&\bullet\ar[r] &\\
&0&&&3&&&0&&2&3&
}\]
\caption{The AHSS for twisted $K$-theory of $\SU(2)$ (left) and $\SO(3)$ (right).
The symbol $\bullet$ denotes a copy of $\bZ$; $\times$ denotes a copy of $\bZ/2$.}
\label{fig:AHSSSU2}
\end{figure}

From the figure, it is apparent that (assuming $h$ and $h\dual$ are both non-zero)
we get $K^\bullet(G,h) = \bZ/h$ in odd
degree, $0$ in even degree.  Similarly, we get $K^\bullet(G\dual,h\dual) = \bZ/h\dual$ 
in odd degree, $\bZ/2$ in even degree.   (This calculation also appears in
\cite[\S3.2]{MR2399311}.)
So with the choice $h=2$ and $h\dual=1$,
we get $\bZ/2$ in odd degree for $G$, $\bZ/2$ in even degree for $G\dual$, 
as is consistent with Corollary \ref{cor:twistedK}.  Note, incidentally, that no other
positive choices for $h$ and $h\dual$ would work here. \qed
\end{example}
\begin{example}[{$\SU(3)$ and $\PSU(3)$}]
  \label{ex:SU3}
Now consider the case of the Langlands dual pair $G=\SU(3)$ and
$G\dual=\PSU(3)$.  We want to compute the appropriate values of $h$ and
$h\dual$ for this case.  Recall from Theorem \ref{thm:maponH3}
that if $\pi\co G\to G\dual$ is the $3$-to-$1$ covering map, then
$\pi^*\co H^3(G\dual)\to H^3(G)$ is an isomorphism in this case.

Since $G$ and $G\dual$ are bundles over $G/T$, we need the cohomology of the
flag variety $G/T$.  This is a classical calculation of Borel
\cite[Proposition 29.2]{MR0051508}:
\[
H^\bullet(G/T) \cong S^\bullet(H^1(T))/S^\bullet(H^1(T))^W,
\]
which in our case reduces to
\begin{multline}
\label{eq:cohomGT}
\bZ[x_1,x_2,x_3]/(x_1+x_2+x_3, \, x_1x_2+x_1x_3+x_2x_3, \, x_1x_2x_3) \\ =
\bZ[x_1,x_2]/(x_1^2+x_2^2+x_1x_2,x_1x_2^2 + x_1^2x_2).
\end{multline}
Here $x_1$, $x_2$, and $x_3$ represent a basis for the weight lattice
of the maximal torus of
$\U(3)$: $x_j(\diag(z_1,z_2,z_3)) = z_j$, and the Weyl group invariants
are generated by the elementary symmetric functions.  Imposing the relation
$x_1+x_2+x_3=0$ is equivalent to restricting to $\SU(3)$.  Inside the
weight lattice $H^1(T)=\bZ x_1 + \bZ x_2$ we have the root lattice
$H^1(T\dual)$, spanned by the simple roots $\alpha =x_1-x_2$ and $x_2-x_3 =
x_1 + 2x_2 = \alpha + 3 x_2$.  So the index of $H^1(T\dual)$ in $H^1(T)$ is
\[
\det \begin{pmatrix} 1 & 1\\ -1 & 2\end{pmatrix} = 3,
\]
as it should be (since $3$ is the degree of the covering $G\to G\dual$).

We would now like to compute $h$ and $h\dual$ explicitly, as multiples
of the generators of $H^3(G)$ and $H^3(G\dual)$.  For this we can split the
torus bundles $G\to G/T$ and $G\dual \to G/T$ into circle bundles and
use \cite[Theorem 5.6]{MR2222224}.  There are many ways to do this,
corresponding to different choices of basis in the weight lattice
$H^1(T)$, but it is convenient to split the bundles as shown:
\[
\xymatrix{ & G\ar[dl]_{\pi_1} \ar[dr]^{\pi_2} & &&&
  G\dual\ar[dl]_{\pi_1\dual} \ar[dr]^{\pi_2\dual} & \\
  X_1 \ar[dr]_{p_1} & & X_2 \ar[dl]^{p_2}&&
  X_1\dual \ar[dr]_{p_1\dual} & & X_2\dual \ar[dl]^{p_2\dual}\\
  & G/T  & &&& G\dual/T\dual = G/T &}
\]
with $c(p_1)=x_1-x_2$, $c(p_2)=x_2$, $c(p_1\dual) = x_1-x_2$,
$c(p_2\dual)=3x_2$.  Note that in this case $X_1=X_1\dual$,
since these are both circle bundles over $G/T$ with the same
Chern class $x_1-x_2$. \cite[Theorem 5.6]{MR2222224} implies
(with our current notation) that
\begin{equation}
\label{eq:Chernclasses}
(\pi_2)_!(h) = p_2^*(c(p_1\dual) )=p_2^*(x_1-x_2), \quad
(\pi_2\dual)_!(h\dual) = (p_2\dual)^*(c(p_1) ) = (p_2\dual)^*(x_1-x_2).
\end{equation}
A simple calculation with the Serre spectral sequence (or the
Gysin sequence, which is basically the same thing) shows that
the cohomology ring of $X_2$ is torsion-free, and is an exterior
algebra on generators in dimensions $2$ and $5$.  Furthermore,
$p_2^*(x_2)=0$ and $p_2^*(x_1)$ generates $H^2(X_2)$.
So from the first equality in \eqref{eq:Chernclasses},
$(\pi_2)_!(h)$ generates $H^2(X_2)$.  It follows that
$h$ has to generate $H^3(G)$, i.e., $h=1$ (with our
usual identifications).  Next we look at $X_2\dual$.  This time,
since $c(p_2\dual)= 3x_2$,
\[
H^2(X_2\dual)=\bZ \left((p_2\dual)^*(x_1) \right)
\oplus (\bZ/3)\left( (p_2\dual)^*(x_2)\right).
\]
From \eqref{eq:Chernclasses}, $(\pi_2\dual)_!(h\dual)$ is still
primitive, so we also have $h\dual = 1$.  Note that
$\pi^*(h\dual) = h$ for the covering $\pi\co G\to G\dual$.

It remains to compute the twisted $K$-theory for $(G,h)$ and
for $(G\dual, h\dual)$ and to check the result of Corollary
\ref{cor:twistedK}.  For $G$ this is easy --- either we can appeal
to \cite{MR2061550} or \cite{MR2263220}, or else we can
compute directly from the AHSS, since $H^*(G)=\bigwedge(x_3,x_5)$
and $h=x_3$, $\Sq^3x_j=0$, so that $d_3$ is multiplication by $x_3$,
which sends $x_3$ to $0$ and $x_5$ to $x_3x_5$.  So
$\ker d_3 = \bZ x_3 + \bZ x_3x_5 = \image d_3$ and $E_4=0$.
Thus $K^\bullet(G,h)=0$ (in both even and odd degree).
In the case of $G\dual$, the cohomology is a bit more
complicated because (via \eqref{eq:cohomPSU}) there is
also a $3$-torsion generator $y$ in degree $2$ with $y^2\ne 0$,
$y^3=0$.  The $5$-dimensional torsion-free generator $x_5$
reduces mod $3$ to $x_1y^2$.  Rationally, things are just as
for $G$, so it suffices to see what happens mod $3$.  The
class $h\dual$ mod $3$ is $x_3$, in the notation of
\eqref{eq:cohomPSU}.  A basis for $H^\bullet(G\dual, \bF_3)$
is given by $1$, $x_1$, $y$, $x_1y$, $x_3$, $y^2$, $x_1x_3$,
$x_3 y$, $x_1 y^2$, $x_1x_3y$, $x_3y^2$, and $x_1x_3y^2$.
Multiplication by $x_3$ kills the $6$ of these monomials
containing an $x_3$ factor, and sends the other $6$ monomials
to these $6$.  So once again we see that $\ker d_3=\image d_3$
and the twisted $K$-theory $K^\bullet(G\dual,
h\dual)=0$ (in both even and odd degree).  This is consistent
with Corollary \ref{cor:twistedK}.  
\end{example}
\begin{example}[{$\Sp(2)=\Spin(5)$ and $\PSp(2)=\SO(5)$}]
  \label{ex:Sp2}
We can do a similar calculation for the Langlands dual pair $G=\Sp(2)$ and
$G\dual=\PSp(2)$, and compute the appropriate values of $h$ and
$h\dual$ for this case.  Recall from Theorem \ref{thm:maponH3}
that if $\pi\co G\to G\dual$ is the $2$-to-$1$ covering map, then
$\pi^*\co H^3(G\dual)\to H^3(G)$ is an isomorphism in this case.

We can take the weight lattice $H^1(T)$ to be $\bZ x_1 \oplus \bZ x_2$,
with the root lattice $H^1(T\dual)$ spanned by the simple
roots $x_1-x_2$ and $2x_2$.  The Weyl group $W$ is the dihedral
group of order $8$ generated by the reflection
$x_1\mapsto x_1, x_2\mapsto -x_2$ and by the rotation
$x_1\mapsto x_2, x_2\mapsto -x_1$.  The ring of $W$-invariants in
$\bZ[x_1,x_2]$ is generated by $x_1^2+x_2^2$ and by $x_1^2x_2^2$,
and $H^\bullet(G/T)\cong \bZ[x_1,x_2]/(x_1^2+x_2^2,x_1^2x_2^2)$,
here with $x_j$ of degree $2$ in the cohomology ring.

In order to apply \cite[Theorem 5.6]{MR2222224},
we again split the
torus bundles $G\to G/T$ and $G\dual \to G/T$  as:
\[
\xymatrix{ & G\ar[dl]_{\pi_1} \ar[dr]^{\pi_2} & &&&
  G\dual\ar[dl]_{\pi_1\dual} \ar[dr]^{\pi_2\dual} & \\
  X_1 \ar[dr]_{p_1} & & X_2 \ar[dl]^{p_2}&&
  X_1\dual \ar[dr]_{p_1\dual} & & X_2\dual \ar[dl]^{p_2\dual}\\
  & G/T  & &&& G\dual/T\dual = G/T &}
\]
with $c(p_1)=x_1-x_2$, $c(p_2)=x_2$, $c(p_1\dual) = x_1-x_2$,
$c(p_2\dual)=2x_2$.  Note that $X_1=X_1\dual$,
since these are both circle bundles over $G/T$ with the same
Chern class $x_1-x_2$. As in Example \ref{ex:SU3},
the cohomology ring of $X_2$ is torsion-free, and is an exterior
algebra on generators in dimensions $2$ and $7$.  Furthermore,
$p_2^*(x_2)=0$ and $p_2^*(x_1)$ generates $H^2(X_2)$.
So from the first equality in \eqref{eq:Chernclasses},
$(\pi_2)_!(h)$ generates $H^2(X_2)$.  It follows that
$h$ has to generate $H^3(G)$, i.e., $h=1$ (with our
usual identifications).  Next we look at $X_2\dual$.  This time,
since $c(p_2\dual)= 2x_2$,
\[
H^2(X_2\dual)=\bZ \left((p_2\dual)^*(x_1) \right)
\oplus (\bZ/2)\left( (p_2\dual)^*(x_2)\right).
\]
From \eqref{eq:Chernclasses}, $(\pi_2\dual)_!(h\dual)$ is still
primitive, so we also have $h\dual = 1$.  Note that
$\pi^*(h\dual) = h$ for the covering $\pi\co G\to G\dual$.

As in the last example, one can check that the twisted
$K$-theory vanishes on both sides of the duality.  (A quick
way of proving this will be given in the following Section
\ref{sec:twistedK}.)
\end{example}  

The phenomena that showed up in Examples \ref{ex:SU2}, \ref{ex:SU3}, 
and \ref{ex:Sp2} can all be explained by the following theorem, which
also explains why we were interested in Theorem \ref{thm:maponH3}.

\begin{theorem}
\label{thm:Tdualclasses}
Let $G$ be a compact simply connected simple Lie group, not of type
$B_n$ or $C_n$ with $n\ge 3$, so that $G\dual$, the corresponding
Langlands dual, is of adjoint type, and there is a covering map
$\pi\co G\to G\dual$.  Then for the Langlands dual pair $G$ and
$G\dual$, with $h$ and $h\dual$ the classes of Theorem 
\textup{\ref{thm:BN1}}, $h\dual$ is the generator of $H^3(G\dual)$
specified by the usual orientation of an embedded rank-$1$ subgroup, i.e.,
it is $1$ when we identify $H^3(G\dual)$ with $\bZ$, and
$h = \pi^*(h\dual)$ as computed in Theorems
\textup{\ref{thm:maponH3}} and \textup{\ref{thm:PSO2n}}.\footnote{In
the exceptional case where $G$ is of type $D_n$ with $n\equiv 0\pmod 4$,
$h\dual$ is still the generator of an infinite cyclic summand
in $H^3(G\dual)$.} Furthermore,
the twisted $K$-theory groups $K^\bullet(G, h)$ and
$K^\bullet(G\dual, h\dual)$ vanish unless $G=\SU(2)$ and
$G\dual = \SO(3)$.
\end{theorem}
\begin{proof}
The key here is to observe that the bundle $p\co G\to G/T$
and the bundle $p\dual\co G\dual\to G\dual/T\dual$ are related
by the covering map $\pi: G\to G\dual$.  Thus if $E$ is the Poincar\'e
bundle (the pull-back of $p$ and $p\dual$), we have a commuting diagram
\begin{equation}
\label{eq:bundles}
\xymatrix{
& E \ar[dl]_{p^*(p\dual)} \ar[dr]^{(p\dual)^*(p)} &\\
G \ar[rr]^\pi \ar[dr]_p& & G\dual\ar[dl]^{p\dual} \\
& G/T \,&.
}
\end{equation}
First, look at the classes $h$ and $h\dual$.  They both pull back to
the same thing in $H^3(E)$, so
\[
({p^*(p\dual)})^*h = ({(p\dual)^*(p)})^*(h\dual).
\]
By commutativity of \eqref{eq:bundles}, $(p\dual)^*(p)
= \pi\circ p^*(p\dual)$, so that means
that
\[
\left({p^*(p\dual)}\right)^* \left(h - \pi^*h\dual\right) = 0,
\]
or $h - \pi^*h\dual =0 $ modulo the kernel of 
$\left({p^*(p\dual)}\right)^*$.  But
\[
c(p^*(p\dual)) = p^*(c(p\dual)) = p^*(c(p)\circ \pi).
\]
However, from the Serre spectral sequence of the $T$-bundle
$p$, $p^*$ kills the image of $c(p)\co H^1(T)\to H^2(G/T)$.
So the Chern class of $p^*(p\dual)$ vanishes, i.e., this is a 
trivial bundle, and hence the pull-back map on cohomology is
faithful.  Hence $h= \pi^*(h\dual)$.

Next, we show that $h\dual$ has to be primitive, i.e., has to
be $\pm 1$ as a multiple of the generator of $H^3(G\dual)$.
(The multiple is positive because of the normalization in
Theorem \ref{thm:BN1}, so it's therefore exactly $1$.)
This follows from the simple connectivity of $G$ together with
the following Lemma \ref{lem:endR}.  Indeed, suppose
$h\dual = k\alpha$ with $\alpha\in H^3(G\dual)$, $k>1$.  Then
by Lemma \ref{lem:endR}, the classifying map for the
T-duality data for $(G\dual\to G/T, h\dual)$
factors as $f_k\circ (\text{something})\co G/T\to \tR$, and
so by the Lemma the classifying map $c(p)\co G/T\to BT$ for
the T-dual is divisible by $k$.  But this would imply that $G$ has
a $k$-fold covering, contradicting the assumption that $G$ is
$1$-connected.

Finally, putting what we have done so far together
with Theorems \ref{thm:maponH3} and \ref{thm:PSO2n}, we see that
$h=1$ or $2$.  By the results of \cite{MR2061550,MR2263220},
that implies that $K^\bullet(G,h)=0$ unless $G=\SU(2)$.  (In all
other cases, we get a direct sum of copies of $\bZ/h'$, where
$h' = \frac{h}{\gcd(h, m)}$ and $m$ is divisible by $2$, so that
$h'=1$.)  By Corollary \ref{cor:twistedK}, we also get
$K^\bullet(G\dual,h\dual)=0$.
\end{proof}
\begin{lemma}
Let $k>1$ be a positive integer and let
$\tR$ be the classifying space for T-duality data for $\bT^n$-bundles
over simply connected
spaces, as defined in \cite{MR2222224}; recall that this is a fibration
as diagrammed in \textup{\eqref{eq:classsp}}.  Then, up to homotopy, $\tR$
has a unique endomorphism $f_k$ which induces multiplication by $k$ on
$\pi_3(\tR)\cong \bZ$ and which on $\pi_2(\tR)\cong \bZ^n \times \bZ^n$
is the identity on the first factor and multiplication by $k$ on the second
factor.  The effect of $f_k$ on T-duality data $(E\xrightarrow{p} Z, h)$
is to keep the bundle $p$
the same and to multiply the H-flux $h$ by $k$.  On the T-dual data, $f_k$
multiplies the characteristic class $c(p\dual)$ by $k$.
\label{lem:endR}
\end{lemma}
\begin{proof}  
This is a simple exercise in obstruction theory, based on the
observation that the given maps on homotopy groups are compatible
with the $k$-invariant of the bundle \textup{\eqref{eq:classsp}}.
\end{proof}  
Theorem \ref{thm:Tdualclasses}
in a sense  is disappointing, in that we get no nontrivial
isomorphisms of twisted $K$-groups, but it completely settles
the nature of the T-dualities coming from Langlands duality.

\subsection{Twisted $K$-theory of compact simple Lie groups}
\label{sec:twistedK}

We now move on to the question of how to compute the twisted
$K$-groups $K^*(G,h)$ more generally, where $G$ is a connected compact
simple Lie group and $h\in H^3(G)\cong \bZ$.  (Strictly speaking,
a twist for $K$-theory is not exactly the same as a class in $H^3$,
but the difference won't matter for our purposes since we just want
to compute the groups as abstract groups.)  Methods for
computing $K^*(G,h)$ were discussed in \cite{MR1877986,MR2079376}
in a few cases, and the calculation for simply connected
$G$ was done completely in \cite{MR2061550,MR2263220}.
However, the methods used there do not work when $G$ is not simply
connected, at least not without a lot of additional work, so
other techniques are needed.  We illustrate another method
using the Segal spectral sequence (from 
\cite[Proposition 5.2]{MR0232393}).  We begin with an abstract
result, which could be made even more general, though
the case given is sufficient for our purposes

\begin{theorem}
Let $F\xrightarrow{\iota} E\to B$ be a fiber bundle, say of
compact metrizable spaces with finite
homotopy type,
and let $h\in H^3(E)$.  Then there is a spectral
sequence $H^p(B, K^q(F, \iota^* h)) \Rightarrow
K^\bullet(E, h)$.
\label{thm:SegalSS}
\end{theorem}
\begin{proof}
In the absence of the twist, this is precisely the
spectral sequence of \cite[Proposition 5.2]{MR0232393} in the
case where the cohomology theory used is complex $K$-theory,
and if $E=B$ and $F=\pt$, it reduces to the usual AHSS.  Similarly,
if $E=B$ and $F=\pt$, but $h\ne 0$, this is the AHSS for twisted
$K$-theory.  To get the general case, we can assume (by
homotopy invariance) that $B$ is a finite CW-complex, and filter
$B$ by its skeleta.  This induces a filtration of $K^\bullet(E, h)$
for which this is the induced spectral sequence (by Segal's proof).

There is also another way to see this in terms of continuous-trace
algebras.  $K^\bullet(E, h)$ is the $K$-theory of a continuous-trace
$C^*$-algebra $A$ with spectrum $B$ and Dixmier-Douady class
$h$.  Using the bundle projection, we can write $A$ as the
algebra of sections of a bundle of $C^*$-algebras over $B$,
where the fiber algebra over a point $b\in B$ is a continuous-trace
algebra with spectrum $E_b\cong F$ and the Dixmier-Douady class
is the pull-back of $h$ to this fiber.  Projections in the fiber
algebras locally extend to a neighborhood of the fiber, and the
differentials of the spectral sequence measure the obstructions
to extending them over the inverse images of bigger and
bigger skeleta of $B$.
\end{proof}

Next let us try to explain the somewhat puzzling results of
\cite{MR2061550} and \cite{MR2263220}, which, to pick just the simplest
case, say that $K^\bullet(\SU(n+1), h)$ is isomorphic to
$\bZ/h'$ tensored with an exterior algebra on $n-1$ generators,
where
\begin{equation}
\label{eq:BraunDouglas}
h' = \frac{h}{\gcd(h, \lcm(1,2,\cdots,n))}.
\end{equation}
The method of proof of this result was quite indirect ---
\cite{MR2061550}  used the Hodgkin K\"unneth spectral sequence
in equivariant $K$-theory together with the calculations
of Freed-Hopkins-Teleman \cite{MR1829086,MR2365650}\footnote{There
is indirect physics input here since Freed-Hopkins-Teleman showed
that the \emph{equivariant} twisted $K$-theory is the same
as the Verlinde ring of the associated WZW model.},
while \cite{MR2263220} used a Rothenberg-Steenrod spectral sequence
and $K$-theory of loop spaces.

Theorem \ref{thm:SegalSS} suggests a potentially much simpler
method for doing the calculations inductively, which we will consider
here for $\SU(n+1)$ and $\PSU(n+1)$, though it could be extended
to other classical groups as well. Let 
$G=\PSU(n+1)$ and $\tG=\SU(n+1)$.  The transitive action of $\tG=\SU(n+1)$ on
the unit sphere $S^{2n+1}$ in $\bC^{n+1}$ descends to a transitive action of $G$
on $L^{2n+1}(n+1)$, the lens space obtained by dividing $S^{2n+1}$ by the action
of $\mu_{n+1}$, the $(n+1)$-th roots of unity.  We thus get a
commuting diagram of fibrations
\begin{equation}
\xymatrix{\SU(n)\ar[r] \ar@{=}@<1ex>[d]& \SU(n+1) \ar[r] \ar@{->>}[d]^\pi
  & S^{2n+1}\ar@{->>}[d]\\
  \SU(n)\ar[r] & \PSU(n+1)  \ar[r]&  \,L^{2n+1}(n+1).}
\label{eq:fibs}
\end{equation}
The $\SU(n)$ at the bottom is really the image in $\PSU(n+1)$ of
matrices of the form $\begin{pmatrix} \zeta & 0\\ 0& \zeta A\end{pmatrix}$,
$A\in \SU(n)$ and $\zeta \in \mu_{n+1}$, but this can be identified with
$\SU(n)$. The diagram \eqref{eq:fibs} gives rise to a morphism
of Serre spectral sequences 
\[
\xymatrix{H^p(S^{2n+1}, H^q(\SU(n))) \ar@{=>}[r]  &
H^{p+q}(\SU(n+1))  \\
H^p(L^{2n+1}(n+1), H^q(\SU(n)))  \ar@{=>}[r] \ar[u]_{\pi^*} &
H^{p+q}(\PSU(n+1)).\ar[u]_{\pi^*} }
\]
Here $H^\bullet(\SU(n))$ is an exterior algebra
on generators $x_3,x_5,\cdots,x_{2n-1}$ (with the subscript indicating the
degree), so the first sequence collapses just for dimensional reasons.
The second sequence looks as if it could have a nontrivial differential
$d_4\co x_3\mapsto cy^2$, where $y\in H^2(L^{2n+1}(n+1))\cong \bZ/(n+1)$ is the
usual generator. However, the calculations of $H^\bullet(\PSU(n+1), \bF_p)$
by Borel in \cite[Th\'eor\`eme 11.4]{MR0064056} and by Baum and Browder in
\cite{MR0189063} (which already appeared in the proof of Theorem
\ref{thm:maponH3})
show that in many cases (e.g., $n+1$ odd), $d_4(x_3)=0$ in the spectral
sequence for the lower fibration in \eqref{eq:fibs}.  Then
the restriction maps $H^3(G)\to H^3(\SU(n))$ and
$H^3(\tG)\to H^3(\SU(n))$ are both isomorphisms.

Now apply Theorem \ref{thm:SegalSS} to the fibrations of
\eqref{eq:fibs}.  We get spectral sequences
\[
H^p(S^{2n+1}, K^q(\SU(n),h)) \Rightarrow K^{p+q}(\SU(n+1), h)
\]
and (say when $n+1$ is odd, so we can still identify
$H^3(\PSU(n+1))$ with $H^3(\SU(n))$)
\[
H^p(L^{2n+1}(n+1), K^q(\SU(n),h)) \Rightarrow K^{p+q}(\PSU(n+1), h).
\]

Assuming we've inductively proved \eqref{eq:BraunDouglas}
for smaller values of $n$, this now gives information about
$K^\bullet(\SU(n+1),h)$ or $K^\bullet(\PSU(n+1),h)$.
We will show how this can be used to get information about
these with relatively elementary methods, not involving the
techniques used by Braun and Douglas.

For example, while the following is weaker than the results
of \cite{MR2079376,MR1877986} (for $\SU(3)$) and
of \cite{MR2061550,MR2263220}
(for $\SU(n+1)$ in general), it does give some
highly nontrivial information, and the proof is relatively elementary.
\begin{theorem}
\label{thm:twistedKofSU}
Suppose $h\in \bZ\cong H^3(\SU(n+1))$ is relatively prime
to $n!$. Then $K^\bullet(\SU(n+1),h)$ is $\bZ/h$ tensored with
an exterior algebra on $n-1$ odd generators.
\end{theorem}
\begin{proof}
We proceed by induction on $n$.  The case $n=1$ is covered
by Example \ref{ex:SU2}, where we computed that
$K^\bullet(\SU(2),h)\cong \bZ/h$ in odd degree.  So assume
that $n>1$ and that
the result holds for smaller values of $n$, and use the Segal 
spectral sequence (Theorem \ref{thm:SegalSS}) associated to the
fibration
\begin{equation}
\SU(n)\to \SU(n+1) \to S^{2n+1}.
\label{eq:SUfib}
\end{equation}
Since $n>1$, it is clear from the fibration that
the restriction map $H^3(\SU(n+1))\to H^3(\SU(n))$ is an isomorphism.
Also, if $h$ is relatively prime to $n!$, then it is certainly
relatively prime to $(n-1)!$.  So we get a Segal spectral
sequence 
\[
H^p(S^{2n+1}, K^q(\SU(n), h)) \Rightarrow
K^{p+q}(\SU(n+1), h),
\]
and by inductive hypothesis, $K^q(\SU(n), h)
\cong (\bZ/h) \otimes \bigwedge(y_1,\cdots, y_{n-2})$,
with the $\bZ/h$ and the $y$'s all in odd degree.
We just need to show that the spectral sequence
collapses.  There is only one possible differential,
$d_{2n+1}$, related to the homotopical nontriviality
of the fibration.  However, the long exact homotopy 
sequence of the fibration includes the following:
\begin{equation}
\label{eq:homotopyseq}
\pi_{2n+1}(\SU(n+1))\to \pi_{2n+1}(S^{2n+1})
\to \pi_{2n}(\SU(n)) \to \pi_{2n}(\SU(n+1)).
\end{equation}
By \cite[\S8, Theorem 5]{MR0102803},
$\pi_{2n+1}(\SU(n+1))\cong \bZ$,
$\pi_{2n}(\SU(n+1))=0$, and $\pi_{2n}(\SU(n)) \cong
\bZ/n!$.  So substituting back into \eqref{eq:homotopyseq}
we get
\[
\cdots \to \bZ \to \bZ \to \bZ/n! \to 0.
\]
That means that after inverting $n!$ (but not without
doing this), the map
\[
\pi_{2n+1}(\SU(n+1))\to \pi_{2n+1}(S^{2n+1})
\]
splits, i.e., the fibration \eqref{eq:SUfib} splits. (A splitting of
the generator of $\pi_{2n+1}(S^{2n+1})$ would be a map
$S^{2n+1}\to \SU(n+1)$ splitting the fibration.) 
Or equivalently, the class of the bundle \eqref{eq:SUfib} 
is given by a class in 
\[
[S^{2n+1}, B\SU(n)] =
\pi_{2n+1}(B\SU(n))\cong \pi_{2n}(\SU(n))\cong \bZ/n!.
\]
So when $\gcd(h, n!)=1$,
$d_{2n+1}$ vanishes and the conclusion follows.
\end{proof}
\begin{remark}
\label{rem:whydenom}
  Note that Theorem \ref{thm:twistedKofSU}, and especially
  the appeal to Bott's theorem \cite[\S8, Theorem 5]{MR0102803},
  ``explains'' the strange denominator in formula \eqref{eq:BraunDouglas},
  at least to some extent.  We can also handle the
  symplectic groups the same way.
\end{remark}
\begin{theorem}
\label{thm:twistedKofSp}
Suppose $h\in \bZ\cong H^3(\Sp(n))$ is relatively prime
to $(2n+1)!$. Then $K^\bullet(\Sp(n),h))$ is $\bZ/h$ tensored with
an exterior algebra on $n-1$ odd generators.
\end{theorem}
\begin{proof}
This proceeds exactly the same way as Theorem
\ref{thm:twistedKofSU}, using the homotopy sequence of the fibration
\[
\Sp(n-1) \to \Sp(n) \to S^{4n-1}.
\]
By Bott periodicity and stability, $\pi_{4n-1}(\Sp(n))\cong \bZ$
and $\pi_{4n-2}(\Sp(n))=0$.  The group $\pi_{4n-2}(\Sp(n-1))$
is not in the stable range since $4n-2 = 4(n-1) + 2$, but it is computed
in \cite{MR0143216} and \cite{MR0169243}, and turns out to be
$\bZ/(2n+1)!$ or $\bZ/(2(2n+1)!)$, depending on whether
$n$ is odd or even.  Thus we get the long exact homotopy sequence
\[
\left(\bZ\cong \pi_{4n-1}(\Sp(n)) \right) \to \pi_{4n-1}(S^{4n-1})
\to \left( \bZ/(2n+1)!\text{ or }\bZ/(2(2n+1)!) \right) \to 0,
\]
and so the fibration splits after inverting primes $\le 2n+1$.
The rest of the proof is as before.
\end{proof}
Results for other groups can be proved with similar methods.
For example, we have:
\begin{theorem}
\label{thm:twistedKofG2}
Suppose $h\in \bZ\cong H^3(G_2)$ is relatively prime
to $2$, $3$, and $5$. Then $K^\bullet(G_2,h)$ is $\bZ/h$ in
both even and odd degree.
\end{theorem}
\begin{proof}
This proceeds similarly, using the homotopy sequence of the fibration
\cite[Lemme 17.1]{MR0064056}
\[
\SU(2) \to G_2 \to V_{7,2}.
\]
The Stiefel manifold $V_{7,2}$ is $11$-dimensional and
has only one nontrivial homology group below the top
dimension, namely a $\bZ/2$ in dimension $5$.
The restriction map $H^3(G_2)\to H^3(\SU(2))$ is an isomorphism,
so we can apply the Segal spectral sequence
and we get a spectral sequence 
\[
E_2^{p,q} = H^p(V_{7,2}, K^q(\SU(2), h)) \Rightarrow K^\bullet(G_2, h).
\]
Here $K^q(\SU(2), h)\cong \bZ/h$ for $q$ odd and is $0$ for
$q$ even.  If $h$ is odd, $E_2^{p,q}$ is non-zero
only for $p=11$ and $q$ odd, and $d_{11}$ is the only possible
differential.  However, we have the long exact homotopy sequence
\[
\pi_{11}(G_2) \to \pi_{11}(V_{7,2}) \to \pi_{10}(\SU(2)) \to \pi_{10}(G_2),
\]
and since $\pi_{10}(\SU(2)) \cong \bZ/15$, the map
$\pi_{11}(G_2) \to \pi_{11}(V_{7,2})$ is surjective after inverting $3$ and $5$.
Furthermore, after inverting $2$, $V_{7,2}$ becomes homotopy equivalent
to $S^{11}$ by the Hurewicz Theorem modulo 
the Serre class of $2$-primary torsion groups.
So when $h$ is prime to $2$, $3$, and $5$, the spectral sequence
is just as for $S^3\times S^{11}$, 
must collapse, and gives the result.
\end{proof}

Similar techniques can also be used to compute twisted $K$-theory
in some cases for non-simply connected groups.  Here is a representative
example:
\begin{theorem}
\label{thm:twistedKPSU3}
Let $h\in \bZ$, identified with $H^3(\PSU(3))$. Then if
$\gcd(h,3)=1$,
\[
K^\bullet(\PSU(3),h)\cong K^\bullet(\SU(3),h)\cong
\begin{cases}\bZ/h\otimes \textstyle{\bigwedge}(x_1), & h\text{ odd},\\
\bZ/(h/2)\otimes \textstyle{\bigwedge}(x_1), & h\text{ even}.  
\end{cases}
\]
On the other hand, if $h=3$, then both $K^{\text{even}}(\PSU(3),h)$ and
$K^{\text{odd}}(\PSU(3),h)$ are elementary abelian $3$-groups of order $27$.
\end{theorem}
\begin{proof}
Since we've seen (via the diagram \eqref{eq:fibs} and Theorem
\ref{thm:maponH3}(1))
that the restriction map $H^3(\PSU(3))\to H^3(\SU(2))$
(coming from \eqref{eq:fibs}) is an isomorphism, we get from
\eqref{eq:fibs} a Segal spectral sequence
\begin{equation}
H^p(L^{5}(3), K^q(\SU(2),h)) \Rightarrow K^{p+q}(\PSU(3), h).  
\label{eq:twistedKPSU3}
\end{equation}
Here we know that $K^q(\SU(2),h)$ is $\bZ/h$, concentrated in odd degree.
There are now various cases.  If $h$ is prime to $3$, then the lens space
$L^{5}(3)$ looks like $S^5$ mod $h$, and the spectral sequence becomes
the same as for $K^\bullet(\SU(3), h)$, for which we know the answer
by \cite{MR2079376,MR1877986,MR2061550,MR2263220}.

So consider the case where $h=3$.  In this case, $H^\bullet(L^{5}(3),\bF_3)
\cong \bF_3[x_1, y_2]/(x_1^2, y_2^3)$ with $\beta x_1=y_2$.  So
the $E_2$ stage of \eqref{eq:twistedKPSU3} gives groups of order
$27$ in both even and odd degree.  So it will suffice to show
that the spectral sequence collapses at $E_2$. Since
$E_2^{p,q}$ is only non-zero for $q$ odd and for $p\le 5$, the only
possible differentials are $d_3$ and $d_5$.  We consider them one at
at time.  Both differentials would vanish if $\PSU(3)$ looked like
$S^3\times L^{5}(3)$, and have to do with nontriviality of the fibration
in \eqref{eq:fibs}.  So how nontrivial is it?  The homotopy groups
of $\PSU(3)$ are given by $\pi_j(\PSU(3)) \cong \bZ/3, j=1$;
$0, j=2$; $\bZ, j=3$; $0, j=4$; and $\bZ, j=5$.  Since $\PSU(3)$ is a
topological group, it is certainly simple (in the sense of homotopy
theory) and has a Postnikov system
even though it isn't simply connected.  The first nontrivial stage of this
Postnikov system is a fibration $K(\bZ,3) \to X_1 \to K(\bZ/3, 1)$.
The next stage is a fibration $K(\bZ,5) \to X_2 \to X_1$.  The
$k$-invariant of $X_1$ lies in $H^4(K(\bZ/3, 1))\cong \bZ/3$, and
corresponds to the transgression $d_4\co H^3(K(\bZ,3))\to
H^4(K(\bZ/3, 1))$ in the Serre spectral sequence for the fibration
defining $X_1$.  But $H^\bullet(K(\bZ,3),\bF_3)\cong  \bF_3[x_1, y_2]/(x_1^2)$
with $\beta x_1=y_2$, which agrees with the cohomology of $ L^{5}(3)$
up to degree $5$, while by \eqref{eq:cohomPSU}, the element corresponding
to $y_2^2$ is nonzero in $H^4(\PSU(3),\bF_3)$.  Thus $d_4$ has to
vanish and $X_1\simeq K(\bZ,3) \times K(\bZ/3, 1)$.  This means
$\PSU(3)$ closely approximates $K(\bZ,3) \times K(\bZ/3, 1)$ through
dimension $4$.  Hence as far as $d_3$ in the Segal spectral sequence
is concerned, we might as well have $\PSU(3)\simeq S^3\times L^{5}(3)$,
which causes $d_3$ to vanish.  Vanishing of $d_5$ then comes
from comparison of spectral sequences; the
diagram \eqref{eq:fibs} plus the fact that there
are no differentials before $d_5$ implies that $d_5$ for
$K^\bullet(\PSU(3),h)$ must agree with $d_5$ for $K^\bullet(\SU(3),h)$,
which we know must vanish if $h$ is odd.  So the twisted $K$-groups
in this case are of order exactly $27$.

Note that even though the Segal spectral sequence degenerates, there
is still an extension problem in going from $E_\infty$ to the actual
twisted $K$-groups.  In the case $h=3$,
one can solve the extension problem by noting that the method of
proof really showed something stronger, namely that
\[
K^\bullet(\PSU(3),3)\cong K^\bullet(\SU(2)\times L^5(3),3)
\cong K^{\bullet+1}(L^5(3), \bF_3).
\]
This turns out to be elementary abelian, for two reasons.
First, the universal coefficient theorem for $K$-theory with
$\bF_p$ coefficients splits, unlike what happens for
$KO$-theory with $\bF_2$ coefficients. (See \cite[\S2]{MR0182967}
for an explanation.)  And secondly, the complex $K$-theory
$\widetilde K^0(L^5(3))$ turns out to be elementary abelian, unlike
what happens for certain other lens spaces.  For
an explanation of this, see \cite{MR0198491} and especially the
\emph{Mathematical Reviews} review of this paper by Hirzebruch.
\end{proof}
\begin{remark}
\label{rem:extensions}
The final result that, when $h=3$,
the twisted $K$-groups in both even and odd degree
are $(\bZ/3)^3$ was obtained previously in \cite[(2.20)]{MR2080884} from
the physics perspective of D-brane charges in WZW theories
for $\PSU(3)$.  More complicated calculations for $\PSU(9)$,
where the result is not so simple, appear in \cite{MR2285965}.

Theorem \ref{thm:twistedKPSU3} actually proves more: if $h=3k$ with
$k$ odd and relatively prime to $3$, then $K^\bullet(H, h)\cong
(\bZ/3)^3 \oplus (\bZ/k)$, and if $h=3k$ with
$k$ even and relatively prime to $3$, then $K^\bullet(H, h)\cong
(\bZ/3)^3 \oplus (\bZ/(k/2))$.  This follows from the fact
that the $p$-primary parts of the Segal spectral sequence for
different values of $p$ do not interfere with each other.
\end{remark}
Note that exactly the same method as for Theorem
\ref{thm:twistedKPSU3} proves:
\begin{theorem}
\label{thm:nonsc}
Let $G=\PSU(n+1)$ or $\PSp(n)$, $n\ge 2$, with universal cover
$\tG=\SU(n+1)$ or $\Sp(n)$.  Suppose $\gcd(h,n+1)=1$
{\lp}in the case of $\PSU(n+1)${\rp} or $h$ is odd {\lp}in the case
of $\PSp(n)${\rp}.  Then 
$\pi^*\co K^\bullet(G,h) \to K^\bullet(\tG, \pi^*h)$
is an isomorphism of twisted $K$-theory groups.
\end{theorem}
\begin{proof}
We indicate the details for $\PSp(n)$; the other case is analogous.
Consider the commuting diagram of fibrations \eqref{eq:fibs1},
which gives a morphism of spectral sequences
\[
\xymatrix{H^p(\bR\bP^{4n-1}, K^q(\Sp(n-1), \iota^*h)) 
\ar@{=>}[r] \ar_{\pi^*}[d] & K^{p+q}(\PSp(n), h) \ar_{\pi^*}[d]\\
H^p(S^{4n-1}, K^q(\Sp(n-1), \iota^*\pi^*h)) 
\ar@{=>}[r] & K^{p+q}(\Sp(n), \pi^*h),}
\]
with $\iota^*$ the pull-back to the fiber.  Note that by Theorem
\ref{thm:maponH3}, $\pi^*h=h$ or $2h$.  But in either event,
by the results of \cite{MR2061550,MR2263220},
since $h$ is odd, $K^\bullet(\Sp(n-1), \pi^*h)$ is $\bZ/h'\otimes
\bigwedge(x_1,\cdots x_{n-2})$, with $x_j$ of odd degree,
and with $h'$ \emph{odd}.  (Furthermore, the value of $h'$ is 
independent of whether or not $h$ is multiplied by $2$.) So
\[
\pi^*\co H^p(\bR\bP^{4n-1}, K^q(\Sp(n-1), \iota^*h)) 
\to H^p(S^{4n-1}, K^q(\Sp(n-1), \iota^*\pi^*h))
\]
is an isomorphism for all $p$ and $q$.  The result now follows
by the comparison theorem for spectral sequences.
In the case of $\PSU(n+1)$, we use the diagram \eqref{eq:fibs};
note that $\pi^*=1$ if $n$ is even
and is $2$ if $n$ is odd.  But in that case $h$, being relatively prime
to $n+1$, is odd.   In
all cases, the order of the torsion in $K^\bullet(\SU(n), \iota^*h)$
is prime to the order of the torsion in the cohomology of the lens
space $L^{2n+1}=S^{2n+1}/\mu_{n+1}$.
 So everything works the same way.
\end{proof}

\subsection{Level-rank dualities}
\label{sec:levelrank}
To conclude this section, we finally mention level-rank dualities,
which are certain mysterious relations between WZW theories
that involve swapping a parameter for the group with one for the
Virasoro level.  (See, e.g., \cite{MR1065265,MR1079689,MR1261989,
MR1285882,MR2215500,MR2263322}.)
This is quite a deep subject and does seem to
be connected with twisted $K$-theory calculations, though
not in a totally straightforward way. In \cite[\S6]{MR2061550}
Braun calls this ``level-rank nonduality,'' but this seems to be too harsh,
as even in Braun's ``nonduality'' examples such as
$B_2$ at level $1$ and $E_8$ at level $2$
(with only a few exceptions, such as $G_2$ at level $1$ \cite[\S4]{MR2061550},
where the twisted $K$-theory vanishes), the \emph{order of
the torsion} in the associated twisted $K$-groups is the same,
as it is this order that gives the dimension of the associated fusion
rings, which are equal for the dual theories.

A few examples will illustrate this point.  The most basic case of
level-rank dualities is between WZW theories on $\SU(n)$ at
level $k$ and $\SU(k)$ at level $n$ \cite{MR1079689}.
(This has evolved into the ``strange duality'' conjecture
proved in \cite{MR2289865,MR2350055,MR2520799}.)
While these theories might at first seem to
have nothing to do with one another, there is one immediate 
connection.  The associated twisting class in both
cases is $n+k$ (level plus dual Coxeter number), and if this
is relatively prime to both $(n-1)!$ and $(k-1)!$ (for
example, in the case of $\SU(3)_4$ and $\SU(4)_3$),
it follows immediately from Theorem \ref{thm:twistedKofSU}
that the order of the torsion in the two twisted $K$-groups
is the same.  Thus we have a very quick way of
verifying this in certain cases.  In some cases the duality
descends to the adjoint group (cf.\ \cite[(9.3)]{MR1261989});
the order of the torsion is also the same for
$\PSU(3)_4$ and $\PSU(4)_3$, since the twisting
is by $3+4=7$ in both cases and we can apply
Theorem \ref{thm:nonsc}.

\section{Orientifold dualities for Lie groups with involution}
\label{sec:orient}

This section of the paper was motivated by the observation (see
for example \cite{GaoHori,MR3267662, MR3316647}) that there are interesting
examples of T-dualities between orientifold
string theories on Calabi-Yau manifolds with holomorphic and
anti-holomorphic involutions (in types IIB and IIA, respectively).
Such dualities
of orientifold theories are expected to result in isomorphisms of
(possibly twisted)
$KR$-groups, with a degree shift that depends on the number of
circles in which one dualizes.  If we interpret the term
``Calabi-Yau manifold'' in the broadest sense,
as a complex manifold, not necessarily compact or simply connected, with 
vanishing first Chern class, then complex Lie groups provide plenty of examples,
since their holomorphic tangent bundles are parallelizable (and thus certainly
have vanishing Chern classes).  Furthermore, semi-simple Lie groups are the 
natural setting for WZW models in string theory and
conformal field theory, so they provide an obvious case of interest.
We further specialize to the case where the orientifold involution is a
group automorphism.  While this is not the only possibility, and other
literature on WZW orientifolds
(for example \cite{MR1888950,MR1900127,MR1955442,MR2118829,MR2443297})
deals with the case of group anti-automorphisms
(or twists thereof), the automorphism case certainly
seems like a natural case to consider, and it could be that
orientifold theories on a group $G$ with a group involution $\iota$ and
$G^\iota=H$ are related to ``coset theories'' on $G/H$.

The T-dualities we exhibit here are motivated by examples of ``group
duality'' (such as Matsuki duality) that occur in representation theory.
We begin with a proposition which is basically well-known.
\begin{proposition}
Let $G$ be a connected semi-simple complex Lie group, and let $G_0$
be a real form of $G$ {\lp}so that $G_0=G^\iota$ for some anti-holomorphic
involution $\iota$ of $G${\rp}.  Let $\theta$ be the Cartan involution of
$G_0$, with fixed-point subgroup $K_0=G_0^\theta$.
Then $\theta$ can be chosen to commute with $\iota$ and to extend
to a holomorphic involution of $G$ {\lp}which we will again denote
by $\theta${\rp} commuting with $\iota$.
\label{prop:invs}
\end{proposition}  
\begin{proof}
As we mentioned, this is standard; see for example \cite{MR513481}.
\end{proof}
\begin{theorem} Let $G$, $G_0$, $K_0$, $\iota$, and $\theta$ be as
in Proposition \ref{prop:invs}.  Then there is a natural isomorphism
$KR^{\bullet}(G, \iota)\cong KR^{\bullet+2(m-n)}(G, \theta)$, where
$m = \dim G_0 - \dim K_0$, $n=\dim K_0$.
The isomorphism can be viewed as coming from a T-duality of 
orientifold theories for the two involutions.
\label{thm:invduality}
\end{theorem}
\begin{proof}
Let $\fg$ and $\fg_0$ be the Lie algebras of $G$ and $G_0$, respectively,
and let $\fk_0$ be the Lie algebra of $K_0$, $\fk$ its complexification
(the Lie algebra of the complex subgroup $K$ of $G$ given by $K=G^\theta$).
Let $\tau=\theta\iota$; since $\theta$ and $\iota$ commute, this is
also an involution, and it's anti-holomorphic since $\theta$ is holomorphic
and $\iota$ is anti-holomorphic.  Thus $U=G^\tau$ is another real form of
$G$.  Let $\fp$ be the $-1$-eigenspace of $\theta$; note that the 
the $-1$-eigenspace of $\iota$ is $i\fg_0$.  So we have the decomposition
\begin{equation}
\fg= \fg_0 + i\fg_0 = \fk_0 + i\fk_0 + \fp_0 + i\fp_0,
\label{eq:Liealgs}
\end{equation}
where $\fp_0=\fp\cap\fg_0$.  The Cartan decomposition of $G_0$ is
$G_0=K_0\exp(\fp_0)$.  The involution $\tau$ acts by $+1$ on $\fk_0$
and on $i\fp_0$, and by $-1$ on $i\fk_0$ and on $\fp_0$.
The Killing form is negative definite on $\fk_0$ and positive definite
on $\fp_0$, hence is negative definite on $\fu= \fk_0 + i\fp_0$, the
Lie algebra of $U$.  So $U$ is compact; it is the compact real
form of $G$ (which is unique up to inner automorphisms).  So the
Cartan decomposition of $G$ is $G=U\exp(i\fk_0+\fp_0)$.

With all these preliminaries, we can now understand the involutions
$\theta$ and $\iota$.  Topologically, $G= U\times (i\fk_0) \times \fp_0$,
where $i\fk_0$ and $\fp_0$ are real vector spaces of dimensions
$\dim K_0$ and $\dim G_0 - \dim K_0$, respectively.  The involutions
$\iota$ and $\theta$ both send $U$ to itself with fixed-point
subgroup $K_0$.  But $\iota$ is $-1$ on $i\fk_0$ and $+1$ on $\fp_0$,
whereas $\theta$ is $+1$ on $i\fk_0$ and $-1$ on $\fp_0$.  So
as Real spaces in the sense of Atiyah \cite{MR0206940}, but with
the opposite indexing convention as in \cite{MR1031992},
$(G,\theta) = (U,\theta)\times \bR^{n,m}$, where 
$n=\dim K_0$ and $m = \dim G_0 - \dim K_0$, while
$(G,\iota) = (U,\theta)\times \bR^{m,n}$.  So
\[
KR^\bullet (G,\theta) \cong KR^{\bullet -m + n} (U,\theta), \quad
KR^\bullet (G,\iota) \cong KR^{\bullet -n + m} (U,\theta).
\]
The result follows.
\end{proof}
\begin{remark}
\label{rem:invH}
Theorem \ref{thm:invduality} may seem too weak since it does not incorporate
the twist in $KR$-theory given by the B-field and H-flux
(explained in \cite[\S4 and \S5]{MR3267662}).  This may be remedied
by observing that the twist is always almost trivial on the Euclidean space
factors $\fk_0$ and $\fp_0$, since they
are contractible. (The Dixmier-Douady class must be trivial, and so the
only way of twisting is by changing the signs of the O-planes, which simply
shifts $KR$ in degree by $4$.)  Then as long as one uses the same twist
on $(U,\theta)$ and $(U,\iota)$, Theorem \ref{thm:invduality}  continues
to hold even in the presence of a topologically nontrivial B-field.

Another comment is that it might seem surprising at first that the
degree shift in Theorem \ref{thm:invduality} is always even, unlike the
situation for torus orientifolds in \cite{GaoHori,MR3267662, MR3316647}, 
where switching between type IIA and type IIB via a single T-duality
involves a shift of $1$.  This can be explained by the fact that here
we are dealing with a slightly different notion of T-duality ---
it is still ``target space'' duality, but we no longer have torus
bundles.  The simplest case of this new kind of 
T-duality is the duality given by the
Fourier transform, between $\bR^{1,0}$ and $\bR^{0,1}$.
(The switch in the involution comes from the fact that
the Fourier transform of a real-valued function satisfies
$\overline{f(x)} = f(-x)$, not $\overline{f(x)} = f(x)$.)  For such T-dualities
of real or complex vector spaces, the degree shift in $KR$-theory is always even.
\end{remark}

\bibliographystyle{amsplain}
\bibliography{groupdualities-2017}
\end{document}